\RequirePackage{etoolbox}
\csdef{input@path}{{style/}{graphics/}}
\documentclass[MSNbibl,number,citesort,seceqn,dvips]{arxbj}
\makeatletter
   \@ifpackageloaded{graphicx}{}{\usepackage{graphicx}}
\makeatother
\usepackage{mathbh}
\usepackage{graphicx}

\aid{0}
\volume{21}
\issue{2}
\pubyear{2015}
\firstpage{881}
\lastpage{908}
\doi{10.3150/13-BEJ591} 

\makeatletter
\renewcommand{\emptyset}{\varnothing}
\newcommand{\mathds}{\mathbb}
\newcommand{\uppi}{\pi}
\renewcommand{\widehat}{\hat}

\newproclaim{definition}{Definition}[section]
\newtheorem{them}{Theorem}[section]
\newremark{ex}{Example}[section]
\newremark{rem}{Remark}[section]
\newtheorem{corollary}{Corollary}[section]
\newtheorem{lemma}{Lemma}[section]
\newcommand{\sgn}{\operatorname{sgn}}
\makeatother

\begin{document}
\begin{frontmatter}

\title{On the almost sure convergence of adaptive allocation procedures}
\runtitle{Convergence of adaptive allocation procedures}

\begin{aug}
\author[1]{\inits{A.}\fnms{Alessandro} \snm{Baldi Antognini}\thanksref{1}\ead[label=e1]{a.baldi@unibo.it}} \and
\author[2]{\inits{M.}\fnms{Maroussa} \snm{Zagoraiou}\corref{}\thanksref{2}\ead[label=e2]{maroussa.zagoraiou@unical.it}}
\address[1]{Department of Statistical Sciences, University of Bologna,
Via Belle Arti 41, 40126 Bologna, Italy.\\ \printead{e1}}
\address[2]{Department of Economics, Statistics and Finance,
University of Calabria, 87036 Arcavacata di Rende (CS), Italy.
\printead{e2}}
\end{aug}

\received{\smonth{3} \syear{2013}}
\revised{\smonth{11} \syear{2013}}

%
\begin{abstract}
In this paper, we provide some general convergence results for adaptive
designs for treatment comparison, both in the absence and presence of
covariates. In particular, we demonstrate the almost sure convergence
of the treatment allocation proportion for a vast class of adaptive
procedures, also including designs that have not been formally
investigated but mainly explored through simulations, such as
Atkinson's
optimum biased coin design, Pocock and Simon's minimization method
and some of its generalizations. Even if the large majority of the
proposals in the literature rely on continuous allocation rules, our
results allow to prove via a unique mathematical framework the
convergence of adaptive allocation methods based on both continuous
and discontinuous randomization functions. Although several examples
of earlier works are included in order to enhance the applicability,
our approach provides substantial insight for future suggestions,
especially in the absence of a prefixed target and for designs
characterized by sequences of allocation rules.
\end{abstract}

%
\begin{keyword}
\kwd{Biased Coin Designs}
\kwd{CARA Procedures}
\kwd{minimization methods}
\kwd{Response-Adaptive designs}
\kwd{sequential allocations}
\end{keyword}

\end{frontmatter}

\section{Introduction}

The past five decades have witnessed a sizeable amount of statistical
research on adaptive randomized designs in the context of clinical
trials for treatment comparison.
These are sequential procedures where at each step the accrued
information is used to make decisions about the way of randomizing the
allocation of the next subject.

Starting from the pioneering work of Efron's Biased Coin Design (BCD)
\cite{Efr71}, several authors have suggested adaptive procedures that,
by taking into account at each step only previous assignments, are
aimed at achieving balance between two available treatments (see, e.g.,
\cite{Bag04,Smi84a,Smi84b,Soa83,Wei78}). We shall refer to these as
Assignment-Adaptive methods. Since clinical trials usually involve
additional information on the experimental units, expressed by a set of
important covariates/prognostic factors, Pocock and Simon \cite{Poc75}
and other authors (see, for instance, \cite{Atk82,Baz11,Beg80,Tav74})
proposed Covariate-Adaptive designs. These methods modify the
allocation probabilities
at each step according to the assignments and the characteristics of
previous statistical units, as well
as those of the present subject, in order to ensure balance between the
treatment groups
among covariates for reducing possible sources of heterogeneity.

Motivated by ethical demands, another different viewpoint is the
Response-Adaptive randomization methods. These are allocation rules
introduced with the aim of skewing the assignments towards the
treatment that appears to be superior at each step (see, e.g., \cite
{Atk05a}) or, more generally, of converging to a desired target
allocation of the treatments which combines inferential and ethical
concerns \cite{Bag2010,Tym07}. The above mentioned framework has been
recently extended in order to incorporate covariates, which has led to
the introduction of the so-called Covariate-Adjusted Response-Adaptive
(CARA) procedures, that is, allocation methods that sequentially modify
the treatment assignments
on the basis of earlier responses and allocations, past covariate
profiles and the characteristics
of the subject under consideration. See \cite{RosCARA01,Zha07} and the
cornerstone book by Hu and Rosenberger \cite{Hu07}.

In general, given a desired target it is possible to adopt different
procedures converging to it, such as the Sequential Maximum Likelihood
design \cite{Mel01}, the Doubly-adaptive BCD \cite{Eis94,Hu04} and
their extensions with covariates given by Zhang \textit{et al.}'s CARA design
\cite{Zha07} and the Covariate-adjusted Doubly-adaptive BCD \cite
{Zha09}, having well established asymptotic properties. However, in the
absence of a given target one of the main problems lies in providing
the asymptotic behaviour of the suggested procedure. This is especially
true in the presence of covariates, where theoretical results seem to
be few and the properties of the suggested procedures have been
explored extensively through simulations; indeed, as stated by
Rosenberger and Sverdlov \cite{Ros08} ``very little theoretical work
has been done in this area, despite the proliferation of papers''. For
instance, even if Pocock and Simon's minimization method is widely used
in the clinical practice, its theoretical properties are still largely
unknown (indeed, Hu and Hu's results \cite{Hu12} do not apply to this
procedure), as well as the properties of several extensions of the
minimization method and of Atkinson's Biased Coin Design \cite{Atk82}.

Moreover, although the large majority of the proposals are based on
continuous and prefixed allocation rules, updated step by step on the
basis of the current allocation proportion and some estimates of the
unknown parameters (usually based on the sufficient statistics of the
model), the recent literature tends to concentrate on discontinuous
randomization functions, such as the Efficient Randomized-Adaptive
Design (ERADE) \cite{Hu09}, because of their low variability.

In this paper, we provide some general convergence results for adaptive
allocation
procedures both in the absence and presence of covariates, continuous
or categorical. By combining the concept of downcrossing (originally
introduced in \cite{Hill80}) and stopping times of stochastic
processes, we demonstrate the almost
sure convergence of the treatment allocation proportion for a large
class of adaptive procedures, even in the absence of a given target,
and thus our approach provides substantial insight for future
suggestions as well as for several existing procedures that have not
been theoretically explored \cite{Her05,Signor93}. In particular, we
prove that Pocock and Simon's minimization method \cite{Poc75} is
asymptotically balanced, both marginally and jointly, showing also the
convergence to balance of Atkinson's BCD \cite{Atk82}. The suggested
approach allow to prove through a unique mathematical framework the
convergence of continuous and discontinuous randomization functions
(like e.g., the Doubly-Adaptive Weighted Differences design \cite
{Ger06}, the Reinforced Doubly-adaptive BCD \cite{Baz12}, ERADE \cite
{Hu09} and Hu and Hu's procedure \cite{Hu12}), taking also into
account designs based on Markov chain structures, such as the
Adjustable BCD \cite{Bag04} and the Covariate-adaptive BCD \cite
{Baz11}, that can be characterized by sequences of allocation rules.
Moreover, by removing some unessential conditions usually assumed in
the literature, our results allow to provide suitable extensions of
several existing procedures.

The paper is structured as follows. Even if Assignment-Adaptive and
Response-Adaptive procedures can be regarded as special cases of CARA
designs, we will treat them separately for the sake of clarity, in
order to describe the general proof scheme in a simple setting, whereas
Covariate-Adaptive methods will be discussed as particular case of CARA rules.
To treat CARA rules in the presence of solely categorical covariates we
need to extend the concept of downcrossing in a vectorial framework;
this generalization is not used for CARA procedures with continuous
prognostic factors and therefore these two cases will be analyzed separately.
Starting from the notation
in Section~\ref{sec2}, Section~\ref{sec3}
deals with Assignment-Adaptive designs, while
Section~\ref{sec4} discusses Response-Adaptive procedures.
Sections~\ref{sec5} and~\ref{sec6} illustrate the asymptotic behavior
of CARA methods in
the case of continuous and categorical covariates, respectively.
Section~\ref{sec7} discusses the relationship between the proposed methodology
and the theory of Stochastic Approximation. To avoid cumbersome
notation, the paper mainly deals with the case of
just two treatments, but the suggested methodology is shown to extend
to more than two (see Section~\ref{sec3.1}).

\section{Notation}\label{sec2}

Suppose that patients come to the trial sequentially and are assigned
to one of two treatments, $A$ and $B$, that we want to compare.
At each step $i\geq1$, a subject will be assigned to one of the
treatments and a response $Y_i$ will be observed. Typically,
the outcome $Y_i$ will depend on the treatment, but it may also depend
on some characteristics
of the subject expressed by a vector $\mathbf{Z}_i$ of
covariates/concomitant variables. We assume that $\{\mathbf{Z}_i\}
_{i\geq1}$ are i.i.d. covariates that are not under the experimenters'
control, but they can be measured before assigning a treatment, and,
conditionally on the treatments and the covariates (if present),
patients' responses are assumed to be independent. Let $\delta_{i}$
denote the $i$th allocation, with $\delta_{i}=1$ if the $i$th subject
is assigned to $A$ and $0$ otherwise; also,
$\widetilde{N}_n=\sum_{i=1}^{n} \delta_{i}$ is the number of
allocations to $A$ after $n$ assignments and $\uppi_n$ the corresponding
proportion, that is, $\uppi_n=n^{-1}\widetilde{N}_n$.

In general, adaptive allocation procedures can be divided in four
different categories according to the experimental information used for
allocating the patients to the treatments. Suppose that the $(n+1)$st
subject is ready to be randomized; if the probability of assigning
treatment $A$ depends on:
\begin{enumerate}[(iii)]
\item[(i)] {\spaceskip=0.2em plus 0.05em minus 0.02emthe past allocations, that is, $\Pr(\delta_{n+1}=1\vert
\delta_1,\ldots,\delta_n)$, we call such a procedure\break
Assignment-}Adaptive (AA);
\item[(ii)] earlier allocations and responses, that is, $\Pr(\delta
_{n+1}=1\vert\delta_1,\ldots,\delta_n; Y_1,\ldots,Y_n)$, then the
design is Response-Adaptive (RA);
\item[(iii)] the previous allocations and covariates, as well as the
covariate of the present subject, that is, $\Pr(\delta_{n+1}=1\vert
\delta_1,\ldots,\delta_n; \mathbf{Z}_1,\ldots,\mathbf
{Z}_{n},\mathbf{Z}_{n+1})$, the procedure is Covariate-Adaptive (CA);
\item[(iv)] the assignments, the outcomes and the covariates of the
previous statistical
units, as well as the characteristics of the current subject that will
be randomized, that is, $\Pr(\delta_{n+1}=1\vert\delta_1,\ldots
,\delta_n; Y_1,\ldots,Y_n; \mathbf{Z}_1,\ldots,\mathbf
{Z}_{n+1})$, then the rule is called Covariate-Adjusted
Response-Adaptive (CARA).
\end{enumerate}
From now on, we will denote with $\Im_n$ the $\sigma$-algebra
representing the natural history of the experiment up to step $n$
associated with a given procedure belonging to each category (with $\Im
_0$ the trivial $\sigma$-field). For instance, in the case of AA
rules, $\Im_n=\sigma \{\delta_1,\ldots,\delta_n \}$,
whereas for RA designs $\Im_n=\sigma \{\delta_1,\ldots,\delta
_n; Y_1,\ldots,Y_n \}$.

The sequence of allocations is a stochastic process and a general way
of representing
it is by the sequence of the conditional probabilities of
assigning treatment $A$ given the past information at every stage, that is,
$\Pr(\delta_{n+1}=1\vert\Im_n)$ for $n\in\mathds{N}$, which is
called the allocation function. Even if the large majority of suggested
procedures assume continuous allocation rules, in this paper, we take
also into account designs with discontinuous randomization functions,
provided that their set of discontinuities is nowhere dense.

\section{Assignment-Adaptive designs}\label{sec3}

Assignment adaptive rules, which depend on the past history
of the experiment only through the sequence of previous allocations,
were proposed as a suitable trade-off between balance (i.e.,
inferential optimality)
and unpredictability in the context of randomized clinical trials.
Indeed, if the main concern is maximum precision of the results (without
ethical demands), as is well-known under the classical linear model
assumptions, balanced design is \emph{universally optimal} \cite
{Sil80}, since it minimizes the most common inferential criteria for
estimation and maximizes the power of the classical statistical tests.
The requirement of balance is considered particularly cogent for phase
III trials, where patients are sequentially enrolled and the total
sample size is often a-priori unknown, so that keeping a reasonable
degree of balance at each step, even for small/moderate samples, is
crucial for stopping the experiment at any time under an excellent
inferential setting.

The simplest sequential randomized procedure that approaches balance is
the completely randomized (CR) design, where every allocation is to
either treatment with probability $1/2$ independently on the previous
steps; thus, $\delta_{1},\delta_{2},\ldots$ are i.i.d. $\operatorname{Be}(1/2)$ so
that, as $n$ tends to infinity, $\uppi_n \rightarrow1/2$ almost surely
from the SLLN for independent r.v.'s. Although CR could represent an
ideal trade-off between balance and unpredictability, this holds only
asymptotically. In fact, CR may generate large imbalances for small
samples, since $n^{-1/2}\uppi_n$ is asymptotically normal, and this may
induce a consistent loss of precision. For this reason, starting from
the pioneering work of Efron \cite{Efr71}, AA rules were introduced in
the literature in order to force the allocations at each step towards
balance maintaining, at the same time, a suitable degree of randomness.

In this section, we shall deal with AA procedures such that
%
%
\begin{equation}
\label{AAdesigns} \Pr(\delta_{n+1}=1\vert\Im_n)=
\varphi^{\mathrm{AA}}(\uppi_n), \qquad\mbox{for } n\geq1,
\end{equation}
where $\varphi^{\mathrm{AA}}\dvtx [0;1] \rightarrow[0;1]$.

%
\begin{definition}\label{DC1}
For any function $\psi\dvtx [0;1] \rightarrow[0;1]$, a point $t \in
[0;1]$ is called a \emph{downcrossing} of $\psi(\cdot)$ if
\[
\forall x<t,\qquad\psi(x) \geq t \quad\mbox{and}\quad\forall x>t,\qquad\psi(x)
\leq t.
\]
\end{definition}

Note that if the function $\psi(x)$ is decreasing, then there exists a
single downcrossing $t\in(0;1)$ and if the equation $\psi(x)=x$ admits
a solution then the downcrossing coincides with it. Clearly, if $\psi
(\cdot)$ is a continuous and decreasing function, then $t$ can be
found directly by solving the equation $\psi(x)=x$.

\begin{them}\label{prop1}
If the allocation function $\varphi^{\mathrm{AA}}(\cdot)$ in (\ref
{AAdesigns}) has a unique downcrossing $t \in(0;1)$, then
$\lim_{n\rightarrow\infty} \uppi_n=t$ a.s.
\end{them}

\begin{pf}
By using a martingale decomposition of the number of assignments to
treatment $A$, we will show that the asymptotic behavior of the
allocation proportion $\uppi_n$ coincides with that of the sequence of
downcrossing points of the corresponding allocation function (i.e., a
constant sequence in the case of AA procedures). The same arguments
will be generalized in the \hyperref[app]{Appendix} to the case of RA
and CARA rules
for random sequences of downcrossings.

At each step $n\geq1$,
%
%
\begin{equation}
\label{SA} \widetilde{N}_n = \sum_{i=1}^{n}
\delta_i=\sum_{i=1}^{n}\bigl\{
\delta _i-E(\delta_{i}| \Im_{i-1})\bigr\}+\sum
_{i=1}^{n}E(\delta_{i}| \Im
_{i-1})=\sum_{i=1}^{n}\Delta
M_i + \sum_{i=1}^{n}
\varphi^{\mathrm
{AA}} (\uppi_{i-1} ),
\end{equation}
where $\Delta M_i=\delta_i-E(\delta_i|\Im_{i-1})$, $\Im_n=\sigma\{
\delta_1,\ldots,\delta_n\}$ and $\uppi_0=0$. Then $\{\Delta M_i;
i\geq1\} $ is a sequence of bounded martingale differences with
$|\Delta M_i |\leq1$ for any $i\geq1$; thus the sequence $\{
M_{n}=\sum_{i=1}^{n}\Delta M_i ; \Im_n \}$ is a martingale with $\sum_{k=1}^{n} E[(\Delta M_i)^2 |\Im_{k-1}] \leq n$, so that as $n$ tends
to infinity $n^{-1}M_n\rightarrow0 $ a.s. Let $l_n=\max \{s\dvt 1
\leq s \leq n, \uppi_s \leq t  \}$, with $\max\emptyset=0$, then
at each step $i>l_n$ we have $\varphi^{\mathrm{AA}}  (\uppi_{i}
)
\leq t$. Note that
\begin{eqnarray*}
\widetilde{N}_n &= &\widetilde{N}_{l_n+1}
+ \sum_{k=l_n+2}^{n} \Delta M_k+
\sum_{k=l_n+2}^{n} E(\delta_{k}|
\Im_{k-1})
\\
&\leq& \widetilde{N}_{l_n} +1 +M_n -M_{l_n+1} +
\sum_{k=l_n+2}^{n} \varphi^{\mathrm{AA}} (
\uppi_{k-1} )
\\
&\leq& \widetilde{N}_{l_n} +1+M_n -M_{l_n+1} + \sum
_{k=l_n+2}^{n} t 
\end{eqnarray*}
and, since $\widetilde{N}_{l_n} \leq l_n t$, then
\[
\widetilde{N}_n - n t\leq M_n -M_{l_n+1} +1-t,
\]
namely
%
%
\begin{equation}
\label{eq1} \uppi_n - t\leq\frac{M_n -M_{l_n+1} +1-t}{n }.
\end{equation}
As $n\rightarrow\infty$, then $l_n\rightarrow\infty$ or $\sup_n
l_n< \infty$, and in either case the r.h.s. of (\ref{eq1}) goes to 0 a.s.
Thus $[\uppi_n - t] ^+ \rightarrow0$ a.s. and,
analogously, $ [(1-\uppi_n) -  (1- t ) ]^+
\rightarrow0$ a.s. Therefore, $\lim_{n\rightarrow\infty} \uppi
_n=t$ a.s.
\end{pf}

%
\begin{ex}
The completely randomized design is defined by letting $\Pr(\delta
_{n+1}=1\vert\Im_n)=1/2$ for every $n$. This corresponds to assume
$\varphi^{\mathrm{CR}}(x)=1/2 $ for all $x\in[0;1]$, which is
continuous and
does not depend on $x$; therefore, $\varphi^{\mathrm{CR}}(\cdot)$
has a single
downcrossing $t=1/2$ and thus $\uppi_n \rightarrow1/2$ a.s. as
$n\rightarrow\infty$. Clearly, this procedure can be naturally
extended to any given desirable target allocation that is a-priori known.
\end{ex}

%
\begin{ex}
Efron's BCD \cite{Efr71} is defined by
\[
\Pr(\delta_{n+1}=1\vert\Im_n)= %
\cases{ p, &
\quad$\mbox{if } D_n<0$,
\cr
1/2, &\quad$\mbox{if }
D_n=0, \qquad\mbox{for } n\geq1$,
\cr
1-p, & \quad$\mbox{if }
D_n>0$, } %
\]
where $D_n=2 \widetilde{N}_n-n$ is the difference between the
allocations to $A$ and $B$ after $n$ steps and $p\in[1/2;1]$ is the
bias parameter. Since $\sgn D_n =\sgn(\uppi_n-1/2)$, then Efron's rule
corresponds to
%
%
\begin{equation}
\label{BCD} \varphi^E(x)= %
\cases{ p, & \quad$\mbox{if }
x<1/2$,
\cr
1/2, & \quad$\mbox{if } x=1/2$,
\cr
1-p, & \quad$\mbox{if } x>1/2$, }
\end{equation}
which has a single downcrossing $t=1/2$ and therefore
$\lim_{n\rightarrow\infty} \uppi_n=1/2$ a.s.
Clearly, Theorem~\ref{prop1} allows to provide suitable extensions of
Efron's coin converging to any given desired target $t^*\in(0;1)$, namely
%
%
\begin{equation}
\label{Efornext} \varphi^{\tilde{E}}(x)= %
\cases{ p_2, &
\quad$\mbox{if } x<t^*$,
\cr
t^*, & \quad$\mbox{if } x=t^*$,
\cr
p_1,
& \quad$\mbox{if } x> t^*$, } %
\end{equation}
where $0\leq p_1 \leq t^* \leq p_2 \leq1$ and at least one of these
inequalities must hold strictly.
\end{ex}

%
\begin{rem}
Note that, from Theorem~\ref{prop1}, for the convergence to a given
desired target~$t^*$:
\begin{enumerate}[(iii)]
\item[(i)] the allocation function should be decreasing; this condition
is quite intuitive, since it corresponds to assume that, at each step,
if the current allocation proportion $\uppi_n$ is greater than $t^*$,
then the next allocation is forced to treatment $B$ with probability
greater than $t^*$ and this probability increases as the difference
$\uppi_n-t^*$ grows;
\item[(ii)] the continuity of the allocation rule is not required and
therefore it is possible to consider discontinuous randomization
functions like, for example, (\ref{BCD}) and (\ref{Efornext});
\item[(iii)] condition $\varphi^{\mathrm{AA}}(t^*)=t^*$ is not requested;
moreover, structures of symmetry of the allocation function are not
needed (e.g., in (\ref{Efornext}) condition $p_2=1-p_1$ is not
required), even if they are typically assumed in order to treat $A$ and
$B$ in the same way. For instance, the following AA procedure
\[
\varphi^{\mathrm{AA}^*}(x)= %
\cases{ 1, & \quad$\mbox{if } x\leq1/2$,
\cr
1/2, & \quad$\mbox{if } x> 1/2$, } %
\]
is asymptotically balanced, that is, $\uppi_n\rightarrow1/2$ a.s. as
$n$ tends to infinity.
\end{enumerate}
\end{rem}

%
\begin{corollary}\label{cor1}
Suppose that $\varphi^{\mathrm{AA}}$ is a composite function such
that $\varphi
^{\mathrm{AA}}(x)=h_1  [h_2  (x ) ]$,
where $h_1\dvtx D\subseteq\mathds{R} \rightarrow[0;1]$ is
decreasing and
$h_2\dvtx[0;1] \rightarrow D$ is continuous and increasing. If $d \in
D$ is
such that $h_1(d)=h_2^{-1}(d)$, then $\lim_{n\rightarrow\infty}
\uppi
_n=h_2^{-1}(d)$ a.s.
\end{corollary}

\begin{pf}
The proof follows easily from Theorem~\ref{prop1}. Indeed, $\varphi
^{\mathrm{AA}}(\cdot)$ is a decreasing function with $\varphi
^{\mathrm{AA}}
[h_2^{-1}(d) ]=h_1(d)=h_2^{-1}(d)$ and therefore $\varphi
^{\mathrm{AA}}(\cdot)$ has a single downcrossing in $h_2^{-1}(d)$.
\end{pf}

%
\begin{ex}\label{weiabcd}
Wei \cite{Wei78} defined his Adaptive BCD by letting
%
%
\begin{equation}
\label{weiallfunc} \Pr ( \delta_{n+1}=1\vert\Im_n ) =
\mathfrak{f} (2\uppi_n-1 ), \qquad\mbox{for } n\geq1,
\end{equation}
where $\mathfrak{f}\dvtx[-1;1] \rightarrow[0;1]$ is a continuous and
decreasing function s.t. $\mathfrak{f}(-x)=1-\mathfrak{f}(x)$.
Set $g(w)=2w-1\dvtx[0;1] \rightarrow[-1;1]$, Wei's allocation
function is
$\varphi^W(x)=\mathfrak{f} [g(x) ]$. Since
$g^{-1}(w)=(w+1)/2 $ for all $w\in[0;1]$, then
$g^{-1}(0)=1/2=\mathfrak{f}(0)$, that is, $1/2$ is the only
downcrossing of $\varphi^W(\cdot)$. Therefore, from Corollary~\ref
{cor1} it follows that $ \uppi_n\rightarrow1/2$ a.s. as $n
\rightarrow\infty$.
\end{ex}

%
\begin{rem}\label{remABCD}
Note that Theorem~\ref{prop1} still holds even if we assume different
randomization functions at each step by letting
$\Pr(\delta_{n+1}=1\vert\Im_n)=\varphi_n^{\mathrm{AA}}(\uppi
_n)$, provided
that $t\in(0;1)$ is the unique downcrossing of $\varphi_n^{\mathrm
{AA}}(\cdot
)$ for every $n\geq1$.
\end{rem}

%
\begin{ex}\label{exABCD}
The Adjustable Biased Coin Design (ABCD) proposed by Baldi Antognini
and Giovagnoli \cite{Bag04} is defined as follows. Let $F(\cdot)\dvtx
\mathds
{R}\rightarrow[0;1]$ be a decreasing function such that
$F(-x)=1-F(x)$, the ABCD assigns the $(n+1)$st subject to treatment $A$
with probability
$\Pr(\delta_{n+1}=1\vert\Im_n) =F(D_n)$, for $n\geq1$.
This corresponds to let
\[
\varphi_n^{\mathrm{ABCD}}(x) =F\bigl[n(2x-1)\bigr],\qquad n\geq1,
\]
and, from the properties of $F(\cdot)$, at each step $n$ the function
$\varphi_n^{\mathrm{ABCD}}(\cdot)$ is decreasing with $\varphi
_n^{\mathrm{ABCD}}
(1/2 )=1/2$. Thus $t=1/2$ is the only downcrossing of $\varphi
_n^{\mathrm{ABCD}}(\cdot)$ for every $n$, so that $\lim_{n
\rightarrow\infty
}\uppi_n= 1/2$ a.s.
\end{ex}

\subsection{The case of several treatments}\label{sec3.1}

Now we briefly discuss AA procedures in the case of several treatments
in order to show how the proposed downcrossing methodology can be
extended to $K> 2$ treatments. Even if the same mathematical structure
could also be applied to the other types of adaptive rules that will be
presented in Sections~\ref{sec4}--\ref{sec6}, we restrict the
presentation of
multi-treatment adaptive procedures only for AA designs, for the sake
of simplicity regarding the notation.

At each step $i\geq1$, let $\delta_{i\jmath}=1$ if the $i$th patient
is assigned to treatment $\jmath$ (with $\jmath=1,\ldots,K$) and 0
otherwise, and set $\boldsymbol{\delta}_i^t=(\delta_{i1},\ldots
,\delta_{iK})$ with $\boldsymbol{\delta}_i^t\boldsymbol{1}_K=1$
(where $\boldsymbol{1}_K$ is the $K$-dim vector of ones). After $n$
steps, let $\widetilde{N}_{n\jmath}=\sum_{i=1}^{n}\delta_{i\jmath
}$ be the number of allocations to treatment $\jmath$ and $\uppi
_{n\jmath}$ the corresponding proportion, i.e. $\uppi_{n\jmath
}=n^{-1}\widetilde{N}_{n\jmath}$; also, set $\widetilde{
\mathbf{N}}_n^t=(\widetilde{N}_{n1},\ldots,\widetilde{N}_{nK})$ and
$\boldsymbol{\uppi}_n^t=(\uppi_{n1},\ldots,\uppi_{nK})$, where
$\widetilde{\mathbf{N}}_n^t\boldsymbol{1}_K=n$ and $\boldsymbol
{\uppi}_n^t\boldsymbol{1}_K=1$.

In this setting, we consider a class of AA designs that assigns the
$(n+1)$st patient to treatment $\jmath$ with probability
%
%
\begin{equation}
 \Pr(\delta_{n+1,\jmath}=1\vert\Im_n)=
\varphi^{\mathrm
{AA}}_{\jmath} (\boldsymbol{\uppi}_n ), \qquad
\mbox{for } n\geq1,
\end{equation}
where $\Im_n=\sigma(\boldsymbol{\delta}_1,\ldots,\boldsymbol
{\delta}_n)$, $\varphi^{\mathrm{AA}}_{\jmath}$ is the allocation
function of
the $\jmath$th treatment and from now on we set $\boldsymbol{\varphi
}^{\mathrm{AA}}(\boldsymbol{\uppi}_n)= (\varphi^{\mathrm{AA}}_{1}
(\boldsymbol{\uppi
}_n  ),\ldots,\varphi^{\mathrm{AA}}_{K} (\boldsymbol
{\uppi}_n
) )$.

%
\begin{definition}\label{DC1bis}
Let $\mathbf{x}= (x_{1},\ldots,x_{K} )$, where $x_{\jmath
}\in[0;1]$ for any $\jmath=1,\ldots,K$, $\psi_{\jmath}(\mathbf
{x})\dvtx\break [0;1]^{K}\rightarrow[0;1]$ and set $\boldsymbol{\psi
}(\mathbf
{x})= (\psi_{1}(\mathbf{x}), \ldots, \psi_{K}(\mathbf
{x}) )$.
Then
$\mathbf{t}=  (t_{1},\ldots,t_{K} )\in[0;1]^K$ is
called a \emph{vectorial} \emph{downcrossing} of $\boldsymbol{\psi
}$ if for any $\jmath=1,\ldots,K$
\[
\mbox{for all } x_{\jmath}<t_{\jmath},\qquad\psi_{\jmath}(
\mathbf{x}) \geq t_{\jmath} \quad\mbox{and}\quad\mbox{for all }
x_{\jmath
}>t_{\jmath
}, \qquad\psi_{\jmath}(\mathbf{x}) \leq
t_{\jmath}.
\]
\end{definition}

Clearly, if $\psi_{\jmath}(\mathbf{x})$ is decreasing in $\mathbf
{x}$ (i.e., componentwise) for any $\jmath$,
then the vectorial downcrossing $\mathbf{t}$ is unique, with
$\mathbf{t}\in(0;1)^ K$; furthermore $\boldsymbol{\psi
}(\mathbf{t})=\mathbf{t}$, provided that the solution exists.

\begin{them}\label{thm1bis}
At each step $n$, suppose that $\varphi^{\mathrm{AA}}_{\jmath}
(\boldsymbol{\uppi}_n  )$ is decreasing in $\boldsymbol{\uppi}_n$
(componentwise) for any $\jmath=1,\ldots,K$, then
$\lim_{n\rightarrow\infty} \boldsymbol{\uppi}_n=\mathbf{t}$ a.s.
\end{them}

\begin{pf}
The proof follows easily from the one in Appendix~\ref{A3}, where $K$
treatments should be considered instead of the strata induced by the
categorical covariates.
\end{pf}

%
\begin{ex}\label{multitr}
In order to achieve balance, that is, $\uppi^{\ast}_{\jmath}=K^{-1}$
for any $\jmath=1,\ldots,K$, Wei \textit{et al.} \cite{Wei86}
considered the
following allocation rules:
%
%
\begin{equation}
\label{weirule1} \Pr(\delta_{n+1,\jmath}=1\vert\Im_n)=
\frac{\uppi_{n\jmath}^{-1}
-1}{\sum_{k=1}^K(\uppi_{nk}^{-1} -1) },
\end{equation}
and
%
%
\begin{equation}
\label{weirule2} \Pr(\delta_{n+1,\jmath}=1\vert\Im_n)=
\frac{1-\uppi_{n\jmath}}{K-1}.
\end{equation}
Both rules are decreasing in $\uppi_{n\jmath}$ ($\jmath=1,\ldots,K$)
and it is straightforward to see that $\mathbf
{t}=K^{-1}\boldsymbol{1}_K$ is the only vectorial downcrossing of the
functions $\boldsymbol{\psi}^{W_1}$ and $\boldsymbol{\psi}^{W_2}$
given by:
\[
\psi^{W_1}_{\jmath} (\mathbf{x} )=\frac{x_\jmath
^{-1} -1}{\sum_{k=1}^K(x_k^{-1} -1) } \quad
\mbox{and}\quad\psi ^{W_2}_{\jmath} (\mathbf{x} )=\frac{1-x_\jmath}{K-1}
\]
and therefore, by Theorem~\ref{thm1bis}, $\lim_{n\rightarrow\infty}
\uppi_{n\jmath} =K^{-1}$ a.s. for any $\jmath=1,\ldots,K$.

Note that, under rule (\ref{weirule2}), $\psi^{W_2}_{\jmath}
(\mathbf{x}  )=\psi^{W_2}_{\jmath} (x_\jmath
)$ (i.e., at each step the allocation probability of each treatment
depends only on the current allocation proportion of that treatment);
in such a case it is sufficient to solve the system of equations $\psi
^{W_2}_{\jmath} (x_\jmath )=x_\jmath$ ($\jmath=1,\ldots,K$).
\end{ex}

\section{Response-Adaptive designs}\label{sec4}

RA rules, which change at each step the allocation probabilities on the
basis of the previous assignments and responses, were originally
introduced as a possible solution to local optimality problems in a
parametric setup, where there exists a desired target allocation
depending on the unknown model parameters \cite{Robb67}. Recently,
they have been also suggested in the context of sequential clinical
trials where ethical purposes are of primary importance, with the aim
of maximizing the power of the test and, simultaneously, skewing the
allocations towards the treatment that appears to be superior (e.g.,
minimizing exposure to the inferior treatment) \cite{Eis94,Ger06,Ros02}.

Suppose that the probability law of the responses under treatments $A$
and $B$ depends on a vector of unknown parameters $\boldsymbol{\gamma
}_A $ and $\boldsymbol{\gamma}_B$, respectively, with $\boldsymbol
{\gamma}^t=(\boldsymbol{\gamma}_A^t,{\boldsymbol{\gamma
}_B^t})\in\Omega$, where $\Omega$ is an open convex subset of
$\mathbb{R} ^{k}$.
Starting with $m$ observations on each treatment, usually assigned by
using restricted randomization, an initial
non-trivial parameter estimation $\widehat{\boldsymbol{\gamma
}}_{2m}$ is derived. Then, at
each step $n\geq2m$ let $\widehat{\boldsymbol{\gamma}}_{n}$
be the estimator of the parameter $\boldsymbol{\gamma}$ based on the
first $n$ observations, which is assumed to be consistent in the i.i.d.
case (i.e., $\lim_{n\rightarrow\infty}\widehat{\boldsymbol{\gamma
}}_{n}= \boldsymbol{\gamma}$ a.s.). Obviously, the speed of
convergence of the allocation proportion is strictly related to the
convergence rate of the chosen estimators; however, their consistency
is sufficient in order to establish the almost sure convergence of
$\uppi_n$.

In this section, we shall deal with RA procedures such that
%
%
\begin{equation}
\label{RAdesigns} \Pr(\delta_{n+1}=1\vert\Im_n)=
\varphi^{\mathrm{RA}} (\uppi _{n} ;\widehat {\boldsymbol{
\gamma}}_{n} ), \qquad\mbox{for } n\geq2m.
\end{equation}
The following definition will help illustrate the asymptotic behaviour
of RA rules and also CARA designs with continuous covariates treated in
Section~\ref{sec5}.

%
\begin{definition}\label{DC2}
Let $\dot{\psi}(x;\mathbf{y})\dvtx [0;1]\times\mathds{R}^d
\rightarrow
[0;1]$. The function $t(\mathbf{y})\dvtx \mathds{R}^d\rightarrow
[0;1] $
is called a \emph{generalized downcrossing} of $\dot{\psi}$ if for
any given $\mathbf{y}\in\mathds{R}^d$ we have
\[
\forall x<t(\mathbf{y}),\qquad\dot{\psi}(x;\mathbf{y}) \geq t(\mathbf {y}) \quad
\mbox{and}\quad\forall x>t(\mathbf{y}),\qquad\dot{\psi }(x;\mathbf {y})\leq t(
\mathbf{y}).
\]
\end{definition}

If the function $\dot{\psi}(x,\mathbf{y})$ is decreasing in $x$,
then the generalized downcrossing $t(\mathbf{y})$ is unique and
$t(\mathbf{y})\neq\{0;1\}$ for any
$\mathbf{y}\in\mathds{R}^d$. Moreover, if there exists a solution of
the equation $\dot{\psi}(x,\mathbf{y})=x$, then $t(\mathbf{y})$
coincides with this solution.

\begin{them}\label{thm2}
Suppose that at each step $n$ the allocation rule $\varphi^{\mathrm
{RA}}
(\uppi_{n} ;\widehat{\boldsymbol{\gamma}}_{n}  )$ is decreasing
in $\uppi_{n}$. If the only generalized downcrossing $t(\widehat
{\boldsymbol{\gamma}}_{n})$ is a continuous function, then
$\lim_{n\rightarrow\infty} \uppi_n=t({\boldsymbol{\gamma}})$ a.s.
\end{them}

\begin{pf}
See Appendix~\ref{A1}.
\end{pf}

%
\begin{ex}
Geraldes \textit{et al.} \cite{Ger06}
introduced the Doubly Adaptive Weighted
Differences Design (DAWD) for binary response trials. Let $\boldsymbol
{\gamma}=(p_A,p_B)^t$ be the vector of the probabilities of success of
$A$ and $B$ and $\widehat{\boldsymbol{\gamma}}_{n}=(\widehat
{p}_{An},\widehat{p}_{Bn})^t$ the corresponding estimate after $n$
steps. When the $(n+1)$st patient is ready to be randomized, the DAWD
allocates him/her to treatment $A$ with probability
%
%
\begin{equation}
\label{dawd} \Pr(\delta_{n+1}=1\vert\Im_n)=\rho
g_1(\widehat{p}_{An}-\widehat {p}_{Bn})+(1-
\rho) g_2 (2 \uppi_n -1 ), \qquad\mbox{for } n\geq2m,
\end{equation}
where $\rho\in[0;1)$ represents an ``ethical weight'' and $g_1,
g_2\dvtx
[-1,1] \rightarrow[0,1]$ are continuous functions s.t.
\begin{enumerate}[(iii)]
\item[(i)] $g_1(0)=g_2(0)=1/2$ and $g_1(1)=g_2(-1)=1$;
\item[(ii)] $g_1(-x)=1-g_1(x)$ and $g_2(-x)=1-g_2(x)$ $\forall x\in[-1;1]$;
\item[(iii)] $g_1(\cdot)$ is non decreasing and $g_2(\cdot)$ is decreasing.
\end{enumerate}
Regarded as a function of $\uppi_{n}$ and $\widehat{\boldsymbol
{\gamma
}}_{n}$, rule (\ref{dawd}) corresponds to
\[
\varphi^{\mathrm{DAWD}} (\uppi_{n} ;\widehat{\boldsymbol {\gamma
}}_{n} ) =\rho g_1\bigl( (1; -1) \widehat{\boldsymbol{
\gamma}}_{n} \bigr)+(1-\rho) g_2 (2 \uppi_n -1
),
\]
which is decreasing in $\uppi_{n}$, so that the equation $\varphi
^{\mathrm{DAWD}}  (\uppi_{n} ;\widehat{\boldsymbol{\gamma
}}_{n}
)=\uppi_{n}$ has a unique solution $t(\widehat{\boldsymbol{\gamma
}}_{n})$, i.e. the generalized downcrossing, which is continuous in
$\widehat{\boldsymbol{\gamma}}_{n}$ (see \cite{Ger06}). Thus
$\lim_{n\rightarrow\infty} \uppi_n=t(\boldsymbol{\gamma})$ a.s.
\end{ex}

Often there is a desired target allocation $\uppi^{\ast}$ to treatment
$A$ that depends on the unknown model parameters, i.e. $\uppi^{\ast
}=\uppi^{\ast}(\boldsymbol{\gamma})$, where $\uppi^{\ast}\dvtx
\Omega
\rightarrow(0;1)$ is a mapping that transforms a $k$-dim vector of
parameters into a scalar one. Thus, Theorem~\ref{thm2} still holds
even if, instead of (\ref{RAdesigns}), we assume
\[
\Pr(\delta_{n+1}=1\vert\Im_n)=\breve{\varphi}^{\mathrm{RA}}
\bigl(\uppi_{n} ;\uppi^{\ast}(\widehat{\boldsymbol{
\gamma}}_{n}) \bigr), \qquad\mbox{for } n\geq2m,
\]
provided that $\uppi^{\ast}(\cdot)$ is a continuous function. In this
case the generalized downcrossing could be more properly denoted by
$t(\widehat{\boldsymbol{\gamma}}_n)=t(\uppi^{\ast}(\widehat
{\boldsymbol{\gamma}}_{n}))$.

%
\begin{ex}
The Doubly-adaptive Biased Coin Design (DBCD) \cite{Eis94,Hu04} is one
of the most effective families of RA procedures aimed at converging to
a desired target $\uppi^{\ast}(\boldsymbol{\gamma})\in(0,1)$ that is
a continuous function of the model parameters. The DBCD assigns
treatment $A$ to the $(n+1)$st subject with probability
%
%
\begin{equation}
\label{dbcd} \Pr(\delta_{n+1}=1\vert\Im_n)=\breve{
\varphi}^{\mathrm
{DBCD}}\bigl(\uppi_n; \uppi ^{\ast}(\widehat{
\boldsymbol{\gamma}}_n)\bigr), \qquad\mbox{for } n\geq2m,
\end{equation}
where the allocation function $\breve{\varphi}$
needs to satisfy the following conditions:
\begin{enumerate}[(iii)]
\item[(i)] $\breve{\varphi}^{\mathrm{DBCD}}(x;y)$ is continuous on
$(0;1)^2$;
\item[(ii)] $\breve{\varphi}^{\mathrm{DBCD}}(x;x)=x$;
\item[(iii)] $\breve{\varphi}^{\mathrm{DBCD}}(x;y)$ is decreasing
in $x$
and increasing in $y$;
\item[(iv)] $\breve{\varphi}^{\mathrm{DBCD}}(x;y)=1-\breve{\varphi
}^{\mathrm{DBCD}}(1-x;1-y)$ for all $x,y\in(0;1)^{2}$.
\end{enumerate}
The {DBCD} forces the allocation proportion to the target
since from conditions (ii) and (iii), when $x>y$ then $\breve{\varphi
}^{\mathrm{DBCD}}(x,y)<y$, whereas
if $x<y$, then $\breve{\varphi}^{\mathrm{DBCD}}(x,y)>y$. However,
condition (i)
is quite restrictive since it does not include several widely-known
proposals based on discontinuous allocation functions, such as Efron's
BCD and its extensions \cite{Hu09}, while condition (iv) simply
guarantees that $A$ and $B$ are treated symmetrically.

Since $\breve{\varphi}^{\mathrm{DBCD}}(x;y)$ is decreasing in $x$
with $\breve
{\varphi}^{\mathrm{DBCD}}(x;x)=x$, then the generalized downcrossing
is unique,
given by
$t(\uppi^{\ast}(\widehat{\boldsymbol{\gamma}}_n))=\uppi^{\ast
}(\widehat{\boldsymbol{\gamma}}_n)$. Thus, from the continuity of
the target $\uppi^{\ast}(\cdot)$ it follows that $\lim_{n\rightarrow
\infty} \uppi_n=\uppi^{\ast}(\boldsymbol{\gamma})$ a.s.
\end{ex}

%
\begin{ex}
In the same spirit of Efron's BCD, Hu, Zhang and He \cite{Hu09} have
recently introduced the ERADE, which is a class of RA procedures based
on discontinuous randomization functions. Let again $\uppi^{\ast
}(\boldsymbol{\gamma})\in(0,1)$ be the desired target, that is
assumed to be a continuous function of the unknown model parameters,
the ERADE assigns treatment $A$ to the $(n+1)$st patient with probability
%
%
\begin{equation}
\label{erade} \Pr(\delta_{n+1}=1\vert\Im_n)= %
\cases{ \alpha\uppi^{\ast}(\widehat{\boldsymbol{\gamma}}_n),
& \quad $\mbox{if } \uppi_n>\uppi^{\ast}(\widehat{
\boldsymbol{\gamma}}_n)$,
\cr
\uppi^{\ast}(\widehat{
\boldsymbol{\gamma}}_n), & \quad$\mbox{if } \uppi _n=
\uppi^{\ast}(\widehat{\boldsymbol{\gamma}}_n)$,
\cr
1- \alpha
\bigl(1-\uppi^{\ast}(\widehat{\boldsymbol{\gamma}}_n)\bigr),
& \quad$\mbox {if } \uppi_n<\uppi^{\ast}(\widehat{
\boldsymbol{\gamma}}_n)$, } %
\end{equation}
where $\alpha\in[0;1)$ governs the degree of randomness. Clearly,
rule (\ref{erade}) corresponds to
\[
\breve{\varphi}^{\mathrm{ERADE}}(x; y)= %
\cases{ \alpha y, & \quad$
\mbox{if } x>y$,
\cr
y, & \quad$\mbox{if } x=y$,
\cr
1- \alpha(1-y), & \quad$
\mbox{if } x<y$, } %
\]
which has a single generalized downcrossing $t(y)=y$; therefore
$\lim_{n\rightarrow\infty} \uppi_n=\uppi^{\ast}(\boldsymbol
{\gamma
})$ a.s.
\end{ex}

%
\begin{rem}
Contrary to the {DBCD} in (\ref{dbcd}) and the ERADE in (\ref{erade}),
from Theorem~\ref{thm2} conditions $\breve{\varphi}^{\mathrm
{RA}}(x;x)=x$ and
$\breve{\varphi}^{\mathrm{RA}}(x;y)=1-\breve{\varphi}^{\mathrm
{RA}}(1-x;1-y)$ are not
requested for guaranteeing the convergence to the chosen target $\uppi
^{\ast}(\boldsymbol{\gamma})$. For instance, if we let
\[
\breve{\varphi}^{\mathrm{RA}} \bigl(\uppi_{n} ;\uppi^{\ast
}(
\widehat {\boldsymbol{\gamma}}_n) \bigr) = %
\cases{
\uppi^{\ast}(\widehat{\boldsymbol{\gamma}}_n)^{\tau},
& \quad$\mbox{if } \uppi_n>\uppi^{\ast}(\widehat{\boldsymbol
{\gamma}}_n)$,
\cr
\uppi^{\ast}(\widehat{\boldsymbol{
\gamma}}_n)^{1/\tau}, & \quad$\mbox{if } \uppi_n
\leq\uppi^{\ast}(\widehat{\boldsymbol {\gamma}}_n)$, }
\]
where the parameter $\tau\geq1$ controls the degree of randomness,
then $\uppi_n\rightarrow\uppi^{\ast}(\boldsymbol{\gamma})$ a.s. as
$n \rightarrow\infty$.
\end{rem}

\section{CARA designs with continuous covariates}\label{sec5}

Since in the actual clinical practice information on patients'
covariates or prognostic factors is usually collected, in some
circumstances it may not be suitable to base the allocation
probabilities only on earlier responses and assignments. This is
particularly true when ethical demands are cogent and the patients have
different profiles that induce heterogeneity in the outcomes.

Starting from the pioneering work of Rosenberger \textit{et al.}
\cite{RosCARA01}, there has been a growing statistical interest
in the topic of CARA randomization procedures. These designs change at
each step the probabilities of allocating treatments by taking into
account all the available
information, namely previous responses, assignments and covariates, as
well as the covariate profile of the current subject, with the aim of
skewing the allocations towards the superior treatment or, in general,
of converging to a desired target allocation depending on the
covariates \cite{Zha07}.

Within this class of procedures, if past outcomes are not taken into
account in the allocation process, then the corresponding class of
rules are called Covariate-Adaptive. The direct application of CA
designs regards clinical trials without ethical demands, where the
experimental aim consists in balancing the assignments of the
treatments across covariates in order to optimize inference \cite{Baz11}.

Due to the fact that the proof scheme for CARA rules with categorical
covariates requires the extension of the concept of downcrossing in a
vectorial framework, which is not used under CARA procedures with
continuous prognostic factors, we will treat these cases separately and
the former will be analyzed in the next section.

From now on, we deal with CARA designs such that
%
%
\begin{equation}
\label{phi general} \Pr(\delta_{n+1}=1\vert\Im_n,
\mathbf{Z}_{n+1}=\mathbf {z}_{n+1})=\varphi^{\mathrm{CARA}}
\bigl(\uppi_{n} ;\widehat {\boldsymbol {\gamma}}_{n},
\mathbf{S}_n,f(\mathbf{z}_{n+1}) \bigr),\qquad n\geq2m,
\end{equation}
where $\Im_n=\sigma(\delta_1,\ldots,\delta_n;Y_1,\ldots
,Y_n;\mathbf{Z}_1,\ldots,\mathbf{Z}_n)$, $f(\cdot)$ is a
known vector function of the covariates of the $(n+1)$st patient
(usually $f$ is the identity function, but it can also incorporate
cross-products to account for interactions among covariates), $\widehat
{\boldsymbol{\gamma}}_{n}$ depends on earlier allocations, covariates
and responses, while $\mathbf{S}_n=\mathbf{S}(\mathbf
{z}_{1},\ldots,\mathbf{z}_{n})$ is a function of the covariates
of the previous patients. In general, it is a vector of sufficient
statistics of the covariate distribution that incorporates the
information on $\mathbf{Z}$ after $n$ steps, and from now on we
always assume that, as $n\rightarrow\infty$,
%
%
\begin{equation}
\label{ipotesicov} \mathbf{S}_n=\mathbf{S}(\mathbf{Z}_{1},
\ldots ,\mathbf{Z}_{n})\rightarrow\boldsymbol{\varsigma} \qquad
\mbox{a.s.}
\end{equation}
Often, $\mathbf{S}_n$ contains the moments up to a given order of
the covariate distribution, and (\ref{ipotesicov}) is satisfied
provided that these moments exist.

\begin{them}\label{thmCARAc}
At each step $n$, suppose that the allocation function $\varphi
^{\mathrm{CARA}}$ in (\ref{phi general}) is decreasing in $\uppi
_{n}$ and let
\[
\tilde{\varphi}_{\mathbf{Z}} (\uppi_{n} ;\widehat {\boldsymbol{
\gamma}}_{n},\mathbf{S}_n )=E_{\mathbf
{Z}_{n+1}} \bigl[
\varphi^{\mathrm{CARA}} \bigl(\uppi_{n} ;\widehat {\boldsymbol {
\gamma}}_{n},\mathbf{S}_n,f(\mathbf{Z}_{n+1})
\bigr) \bigr].
\]
If the only generalized downcrossing $\tilde{t}_{\mathbf
{Z}}(\widehat{\boldsymbol{\gamma}}_{n},\mathbf{S}_n)$ of
$\tilde{\varphi}_{\mathbf{Z}}$ is jointly continuous, then
%
%
\begin{equation}
\label{thmcovcont} \lim_{n\rightarrow\infty} \uppi_n=
\tilde{t}_{\mathbf
{Z}}({\boldsymbol{\gamma}},\boldsymbol{\varsigma}) \qquad
\mbox{a.s.}
\end{equation}
\end{them}

\begin{pf}
See Appendix~\ref{A2}.
\end{pf}

\begin{ex}
Consider the linear homoscedastic model with treatment/covariate
interactions in the following form
\[
E(Y_{i}) =\delta_{i} \mu_{A}+(1-
\delta_{i}) \mu_{B}+ {z}_{i} \bigl[
\delta_{i}\beta_{A}+(1-\delta_{i})
\beta_{B} \bigr],\qquad i\geq1,
\]
where $\mu_{A}$ and $\mu_{B}$ are the baseline treatment effects,
$\beta_{A}\neq\beta_{B}$ are different regression parameters and
$z_{i}$ is a scalar covariate observed on the $i$th individual, which
is assumed to be a standard normal.
Under this model, adopting ``the-larger-the-better'' scenario,
treatment $A$ is the best for patient $(n+1)$ if
$\mu_{A}+z_{n+1} \beta_{A}>\mu_{B}+z_{n+1} \beta_{B}$; thus, if
only ethical aims are taken into account it could be reasonable to
consider the following allocation rule:
%
%
\begin{equation}
\label{ruleetica} \varphi^{\mathrm{ETH}} \bigl(\uppi_{n} ;\widehat{
\boldsymbol {\gamma }}_{n},\mathbf{S}_n,f(
\mathbf{z}_{n+1}) \bigr)= \mathbh{1}_{ \{\hat{\mu}_{An}-\hat{\mu}_{Bn} +z_{n+1}
(\hat{\beta}_{An} -\hat{\beta}_{Bn}  )>0  \}},
\end{equation}
where $\mathbh{1}_{\{\cdot\}}$ is the indicator function and
$\hat{\boldsymbol{\gamma}}_n=(\hat{\mu} _{An},\hat{\mu}
_{Bn},\hat{\beta}_{An},
\hat{\beta}_{Bn})^{t}$ is the least square estimator of $\boldsymbol
{\gamma}=(\mu_{A},\mu_{B},\beta_{A},
\beta_{B})^{t}$ after $n$ steps. Thus,
%
%
\begin{eqnarray}
\label{cont2} &&E_{\mathbf{Z}_{n+1}} \bigl[\varphi^{\mathrm{ETH}} \bigl(
\uppi_{n} ;\widehat{\boldsymbol{\gamma}}_{n},
\mathbf{S}_n,f(\mathbf {Z}_{n+1}) \bigr) \bigr]
\nonumber
\\[-8pt]
\\[-8pt]
&&\quad= \Pr{ \bigl\{\hat{\mu}_{An}-\hat{\mu}_{Bn}
+Z_{n+1} (\hat {\beta}_{An} -\hat{\beta}_{Bn} )>0
\bigr\}} = 1-\Phi \biggl(\frac{\hat{\mu}_{Bn}-\hat{\mu}_{An}}{\vert\hat
{\beta}_{An}-\hat{\beta}_{Bn} \vert} \biggr),
\nonumber
\end{eqnarray}
where $\Phi(\cdot)$ is the cdf of $Z$.
Note that (\ref{cont2}) is constant in $\uppi_{n}$, so it has a single
generalized downcrossing and from Theorem~\ref{thmCARAc},
\[
\lim_{n\rightarrow\infty} \uppi_n=1-\Phi \biggl(
\frac{\mu_{B}-\mu
_{A}}{\vert\beta_{A} -\beta_{B}\vert} \biggr).
\]
Clearly, (\ref{ruleetica}) is a deterministic allocation function that
at each step assigns the treatment that appears to be superior for the
current subject. Excluding degenerate cases, even if both treatments
are explored over the covariate domain (which is due to the random
nature of the covariates), this rule is improper for clinical
applications, since a random component in the assignments is
fundamental and a suitable compromise between ethical demands and
inferential efficiency is usually needed. This dilemma, usually known
in the clinical literature as ``individual versus
collective ethics'' \cite{Bag2010}, corresponds to the trade-off
between ``exploitation'' and ``exploration'' of the Bandits
literature \cite{Auer02,Git79}. Although
Adaptive randomization \cite{Ros01} and Bandits methodology are very
different approaches, since under the latter a deterministic policy
(i.e., a sequence of allocations) is usually selected in a finite time
horizon in order to maximize a total expected reward over all the
possible sequences (often made in a Bayesian setting), similar
conclusions as those of the present example have been recently
developed by Pavlidis \textit{et al.} \cite{Pav08} in the case of Multi-Armed
Bandits with linear reward in the presence of covariates.
\end{ex}

%
\begin{ex}
As in the case of RA procedures, also for CARA rules there is often a
desired target allocation $\uppi^{\ast}$ to treatment $A$ that is a
function of the unknown model parameters and the covariates, that is,
$\uppi^{\ast}=\uppi^{\ast}(\boldsymbol{\gamma},\mathbf{z})$,
which is assumed to be continuous in ${\boldsymbol{\gamma}}$ for any
fixed covariate level $\mathbf{z}$.
In particular, Zhang \textit{et al.} \cite{Zha07} assumed a
generalized linear
model setup and suggested to allocate subject $(n+1)$ to $A$ with probability
%
%
\begin{equation}
\label{CARAhu07} \Pr(\delta_{n+1}=1\vert\Im_n,
\mathbf{Z}_{n+1}=\mathbf {z}_{n+1})=\uppi^{\ast}(
\widehat{\boldsymbol{\gamma}}_n, \mathbf {z}_{n+1} ), \qquad
\mbox{for } n\geq2m,
\end{equation}
which represents an analog of the Sequential Maximum Likelihood design
\cite{Mel01} in the presence of
covariates. Assuming that the target function $\uppi^{\ast}$ is differentiable
in $\boldsymbol{\gamma}$, under the expectation, with bounded
derivatives, the authors showed that
$\lim_{n\rightarrow\infty}\uppi_{n}=E_{\mathbf{Z}}[\uppi^{\ast
}(\boldsymbol{\gamma}, \mathbf{Z} )]$ a.s.

Clearly, allocation rule (\ref{CARAhu07}) is constant in $\uppi_{n}$
and therefore
$\tilde{\varphi}_{\mathbf{Z}}  (\uppi_{n} ;\widehat
{\boldsymbol{\gamma}}_{n},\mathbf{S}_n )=E_{\mathbf
{Z}_{n+1}} [\uppi^{\ast}(\widehat{\boldsymbol{\gamma}}_n,\allowbreak
\mathbf{Z}_{n+1} ) ]$ is also constant in $\uppi_{n}$. Thus,
the generalized downcrossing of $\tilde{\varphi}_{\mathbf{Z}}$
is unique and obviously
$\lim_{n\rightarrow\infty} \uppi_n= E_{\mathbf{Z}} [\uppi
^{\ast}(\boldsymbol{\gamma}, \mathbf{Z} ) ]$ a.s.
\end{ex}

%
\begin{rem}
Some authors (see for instance \cite{Ban01}) suggested CARA designs
that incorporate covariate information
in the randomization process, but ignoring the covariate of the current
subject. Note that these methods can be regarded as special cases of
$\varphi^{\mathrm{CARA}}$ in (\ref{phi general}) and therefore
Theorem~\ref
{thmCARAc} can still be applied by taking into account the generalized
downcrossing of $\varphi^{\mathrm{CARA}}$ directly.
\end{rem}

Even if Theorem~\ref{thmCARAc} proves the convergence of CARA designs
in the case of continuous covariates, it could be difficult to obtain
an analytical expression for $\tilde{\varphi}_{\mathbf{Z}}$ and
therefore to find the corresponding generalized downcrossing.
Nevertheless, the following lemma allows to obtain the generalized
downcrossing in a simple manner in some circumstances.

%
\begin{lemma}\label{lemma1}
Let $\varphi^{\mathrm{CARA}}  (\uppi_{n} ;\widehat{\boldsymbol
{\gamma
}}_{n},\mathbf{S}_n,f(\mathbf{z}_{n+1})  )$ be jointly
continuous and, assuming that  $\varphi^{\mathrm{CARA}} (x
;\boldsymbol
{\gamma},\boldsymbol{\varsigma},f(\mathbf{Z}) )$ is
decreasing in $x$, let $t^{*}_{\mathbf{Z}}(\boldsymbol{\gamma
},\boldsymbol{\varsigma})$ be the unique solution of equation
\[
\varphi^{\mathrm{CARA}} \bigl(x ;\boldsymbol{\gamma},\boldsymbol {\varsigma
},E_{\mathbf{Z}}\bigl[f(\mathbf{Z})\bigr] \bigr)=x.
\]
If $\varphi^{\mathrm{CARA}}  (t_{\mathbf{Z}}^{*}(\boldsymbol
{\gamma
}, \boldsymbol{\varsigma}) ;\boldsymbol{\gamma},\boldsymbol
{\varsigma},f(\mathbf{Z})  )$ is linear in $f(\mathbf
{Z})$ and $t^{*}_{\mathbf{Z}}$ is jointly continuous,
then (\ref{thmcovcont}) still holds with $\tilde{t}_{\mathbf
{Z}}(\boldsymbol{\gamma}, \boldsymbol{\varsigma})=t_{\mathbf
{Z}}^{*}(\boldsymbol{\gamma}, \boldsymbol{\varsigma})$.
\end{lemma}

\begin{pf}
Assume that $\tilde{t}_{\mathbf{Z}}(\boldsymbol{\gamma}
,\boldsymbol{\varsigma})< t_{\mathbf{Z}}^{*}(\boldsymbol{\gamma
}, \boldsymbol{\varsigma})$. From the properties of $\varphi
^{\mathrm{CARA}}$, the function $\tilde{\varphi}_{\mathbf{Z}}  (x
;\boldsymbol{\gamma}, \boldsymbol{\varsigma} )$ is jointly
continuous and decreasing in $x$, so that
$\tilde{t}_{\mathbf{Z}}(\boldsymbol{\gamma}, \boldsymbol
{\varsigma})=\tilde{\varphi}_{\mathbf{Z}}  (\tilde
{t}_{\mathbf{Z}}(\boldsymbol{\gamma}, \boldsymbol{\varsigma})
;\boldsymbol{\gamma},\boldsymbol{\varsigma} )> \tilde
{\varphi}_{\mathbf{Z}}  (t_{\mathbf{Z}}^*(\boldsymbol
{\gamma}, \boldsymbol{\varsigma}) ;\allowbreak  \boldsymbol{\gamma},\boldsymbol
{\varsigma} )$.
However,
\[
\tilde{\varphi}_{\mathbf{Z}} \bigl(t_{\mathbf
{Z}}^*(\boldsymbol{\gamma},
\boldsymbol{\varsigma}) ;\boldsymbol {\gamma},\boldsymbol{\varsigma} \bigr)=
\varphi^{\mathrm{CARA}} \bigl(t_{\mathbf{Z}}^{*}(\boldsymbol {\gamma}
,\boldsymbol{\varsigma}) ;\boldsymbol{\gamma},\boldsymbol {\varsigma},E_{\mathbf{Z}}
\bigl[f(\mathbf{Z})\bigr] \bigr)=t_{\mathbf{Z}}^*(\boldsymbol{\gamma},
\boldsymbol{\varsigma}),
\]
since $\varphi^{\mathrm{CARA}}  (t_{\mathbf{Z}}^{*}(\boldsymbol
{\gamma}, \boldsymbol{\varsigma}) ;\boldsymbol{\gamma},\boldsymbol
{\varsigma},f(\mathbf{Z})  )$ is linear in $f(\mathbf
{Z})$, contradicting the assumption. Analogously if we assume $\tilde
{t}_{\mathbf{Z}}(\boldsymbol{\gamma},\boldsymbol{\varsigma})>
t_{\mathbf{Z}}^{*}(\boldsymbol{\gamma},\boldsymbol{\varsigma})$.
\end{pf}

%
\begin{ex}
The Covariate-adjusted Doubly-adaptive Biased Coin Design introduced by
Zhang and Hu \cite{Zha09}
is a class of CARA procedures intended to converge
to a desired target $\uppi^*(\boldsymbol{\gamma},\mathbf{z})$.
When the $(n+1)$st subject with covariate $\mathbf
{Z}_{n+1}=\mathbf{z}_{n+1}$ is ready to be randomized, he/she will
be assigned to $A$ with probability
%
%
\begin{eqnarray}
\label{cinesacci} %
&&\Pr(\delta_{n+1}= 1\vert\Im_n,
\mathbf{Z}_{n+1}=\mathbf {z}_{n+1})
\nonumber
\\[-8pt]
\\[-8pt]
&&\quad=\frac{\uppi^*(\widehat{\boldsymbol{\gamma}}_n, \mathbf
{z}_{n+1} )  ( {\widehat{\rho}_n}/{\uppi_n}  )^\nu
}{\uppi^*(\widehat{\boldsymbol{\gamma}}_n, \mathbf{z}_{n+1}
) ( {\widehat{\rho}_n}/{\uppi_n}  )^\nu+ [1-\uppi
^*(\widehat{\boldsymbol{\gamma}}_n, \mathbf{z}_{n+1} )]  (
{(1-\widehat{\rho}_n)}/{(1-\uppi_n)}  )^\nu},
\nonumber
\end{eqnarray}
where $\widehat{\rho}_n=n^{-1} \sum_{i=1}^{n}\uppi^*(\widehat
{\boldsymbol{\gamma}}_n,\mathbf{z}_i)$. Assuming that
%
%
\begin{equation}
\label{hpcin} \Pr(\delta_{n+1}=1\vert\Im_n,
\mathbf{Z}_{n+1}=\mathbf{z}) \rightarrow\uppi^*(\boldsymbol{\gamma},
\mathbf{z}) \qquad\mbox{a.s.}
\end{equation}
the authors proved that $\lim_{n\rightarrow\infty} \uppi
_n=E_{\mathbf{Z}} [\uppi^*(\boldsymbol{\gamma},\mathbf
{Z}) ]$ a.s.

Note that rule (\ref{cinesacci}) can be regarded as special case of
$\varphi^{\mathrm{CARA}}$ after the transformation $(\widehat
{\boldsymbol
{\gamma}}_{n},\mathbf{S}_n,f(\mathbf{z}_{n+1}) )\mapsto
(\widehat{\rho}_n,\uppi^*(\widehat{\boldsymbol{\gamma}}_n,
\mathbf{z}_{n+1} ))$
and thus, even if we remove condition (\ref{hpcin}), Lemma~\ref
{lemma1} can be applied to the allocation function
\[
\breve{\varphi}^{\mathrm{ZH}} (x;a,b )= \biggl\{ 1+ \frac
{1-b }{b} \biggl[
\frac{(1-a) x}{a(1-x)} \biggr]^{\nu} \biggr\}^{-1},
\]
which is decreasing in $x$ and continuous in all the arguments.
Indeed, since both $\widehat{\rho}_n$ and $E_{\mathbf
{Z}_{n+1}} [\uppi^*(\widehat{\boldsymbol{\gamma}}_n, \mathbf
{Z}_{n+1} ) ]$ converge to $E_{\mathbf{Z}} [\uppi
^*(\boldsymbol{\gamma},\mathbf{Z}) ]$ a.s., the solution
of the equation  $\breve{\varphi}^{\mathrm{ZH}}(x;   E_{\mathbf
{Z}} [\uppi
^*(\boldsymbol{\gamma},\mathbf{Z}) ],E_{\mathbf
{Z}} [\uppi^*(\boldsymbol{\gamma},\mathbf{Z}) ])
=x$ is $t^{*}_{\mathbf{Z}}= E_{\mathbf{Z}} [\uppi
^*(\boldsymbol{\gamma},\mathbf{Z}) ]$. Furthermore, since
$\breve{\varphi}^{\mathrm{ZH}} (E_{\mathbf{Z}} [\uppi
^*(\boldsymbol{\gamma},\mathbf{Z}) ];\allowbreak  E_{\mathbf
{Z}} [\uppi^*(\boldsymbol{\gamma},\mathbf{Z}) ],
\uppi^*(\boldsymbol{\gamma}, \mathbf{Z}) )= \uppi
^*(\boldsymbol{\gamma}, \mathbf{Z})$,
then $\lim_{n\rightarrow\infty} \uppi_n=E_{\mathbf{Z}} [\uppi
^*(\boldsymbol{\gamma},\mathbf{Z}) ]$ a.s.
\end{ex}

%
\begin{rem}
Theorem~\ref{thmCARAc} and Lemma~\ref{lemma1} can be naturally
applied to CA designs in the presence of continuous covariates by
considering, instead of (\ref{phi general}), the following class of
allocation rules:
\[
\Pr(\delta_{n+1}=1\vert\Im_n, \mathbf{Z}_{n+1}=
\mathbf {z}_{n+1})=\varphi^{\mathrm{CA}} \bigl(\uppi_{n} ;
\mathbf {S}_n,f(\mathbf{z}_{n+1}) \bigr),
\]
with $\Im_n=\sigma(\delta_1,\ldots,\delta_n;\mathbf
{Z}_ {1},\ldots,\mathbf{Z}_{ {n}})$. Clearly, $\tilde{t}_{\mathbf
{Z}}({\boldsymbol{\gamma}},\boldsymbol{\varsigma})$ and
$t^{*}_{\mathbf{Z}}(\boldsymbol{\gamma},\boldsymbol{\varsigma
})$ should be
replaced by $\tilde{t}_{\mathbf{Z}}(\boldsymbol{\varsigma})$
and $t^{*}_{\mathbf{Z}}(\boldsymbol{\varsigma})$, respectively.
\end{rem}

\section{CARA designs with categorical covariates}\label{sec6}

We now provide a convergence result for CARA designs in the case of
categorical covariates. In order to avoid cumbersome
notation, from now on we assume without loss of generality two categorical
covariates, i.e. $\mathbf{Z}=(T,W)$, with levels $t_j$
($j=0,\ldots,J$) and $w_l$ ($l=0,\ldots,L$), respectively.
Also, let $\mathbf{p} = [p_{jl}\dvt j =0,\ldots,J; l = 0, \ldots
,L]$ be the joint probability distribution of the categorical
covariates, with $p_{jl}>0$ for any $j=0,\ldots,J$ and $l=0,\ldots,L$
and $\sum_{j=0}^J\sum_{l=0}^L p_{jl}=1$.

After $n$ steps, let $N_n(j,l)=\sum_{i=1}^{n}\mathbh{1}_{\{
Z_i=(t_j,w_l)\}}$ be the number of subjects within the
stratum $(t_j,w_l)$, $\widetilde{N}_n(j,l)=\sum_{i=1}^{n}\delta_i
\mathbh{1}_{\{Z_i=(t_j,w_l)\}}$ the number of allocations to $A$
within this
stratum and $\uppi_n(j,l)$
the corresponding proportion, that is, $\uppi
_n(j,l)=N_n(j,l)^{-1}\widetilde{N}_n(j,l)$, for any $j=0,\ldots,J$
and $l=0,\ldots,L$. Also, let $\boldsymbol{\uppi}_n= [\uppi
_n(j,l)\dvt j =0,\ldots,J; l = 0, \ldots,L ]$.

After an initial stage with $m$ observations on each treatment,
performed to derive a non-trivial parameter estimation, we consider a
class of CARA designs that assigns the $(n+1)$st patient with covariate
profile $\mathbf{Z}_{n+1}=(t_j,w_l)$ to $A$ with probability
%
%
\begin{equation}
\label{randcaracat} \Pr\bigl(\delta_{n+1}=1\vert\Im_n,
\mathbf{Z}_{n+1}=(t_j,w_l)\bigr)=\varphi
_{jl} (\boldsymbol{\uppi}_n ;\widehat {\boldsymbol{
\gamma}}_{n},\mathbf{S}_{n} ), \qquad\mbox {for } n
\geq2m,
\end{equation}
where $\Im_n=\sigma(\delta_1,\ldots,\delta_n;Y_1,\ldots
,Y_n;\mathbf{Z}_1,\ldots,\mathbf{Z}_n)$ and $\varphi_ {jl}$
is the allocation function of the stratum $(t_j,w_l)$.

Let $\boldsymbol{\varphi}(\boldsymbol{\uppi}_n;\widehat
{\boldsymbol
{\gamma}}_{n},\mathbf{S}_{n} )= [\varphi_{jl}(\boldsymbol{\uppi
}_n;\widehat{\boldsymbol{\gamma}}_{n},\mathbf{S}_{n} )\dvt j
=0,\ldots,J; l = 0, \ldots,L]$,
often the allocation rule at each stratum does not depend on the entire
vector of allocation proportions $\boldsymbol{\uppi}_n$ involving all
the strata, but depends only on the current allocation proportion of
this stratum, that is,
%
%
\begin{equation}
\label{stratrand} \varphi_{jl}(\boldsymbol{\uppi}_n;
\widehat{\boldsymbol{\gamma }}_{n},\mathbf{S}_{n})=
\varphi_{jl}\bigl(\uppi_n(j,l);\widehat {\boldsymbol{
\gamma}}_{n}, \mathbf{S}_{n}\bigr), \qquad\forall j =0,
\ldots ,J; l = 0, \ldots,L.
\end{equation}
However, note that (\ref{stratrand}) does not correspond in general to
a stratified randomization, due to the fact that the estimate $\widehat
{\boldsymbol{\gamma}}_{n}$ usually involves the information accrued
from all the strata up to that step, and thus the evolutions of the
procedure at different strata are not independent.

\begin{definition}\label{DC3}
Let $\mathbf{x}= [x_{1},\ldots,x_{\mathcal{K}} ]$, where
$x_{\iota}\in[0;1]$ for any $\iota=1,\ldots,\mathcal{K}$ and
$\mathcal{K}$ is a positive integer. Also, let $\ddot{\psi}_{\iota
}(\mathbf{x};\mathbf{y})\dvtx [0;1]^{\mathcal{K}}\times\mathds{R}^d
\rightarrow[0;1]$ and set $\ddot{\boldsymbol{\psi}}(\mathbf
{x};\mathbf{y})= [\ddot{\psi}_{1}(\mathbf{x};\mathbf
{y}),\ldots,\ddot{\psi}_{\mathcal{K}}(\mathbf{x};\mathbf
{y}) ]$.
Then
$\mathbf{t}(\mathbf{y})=  [t_{1}(\mathbf{y}),\ldots
,t_{\mathcal{K}}(\mathbf{y}) ]$, with $t_{\iota}(\mathbf
{y})\dvtx
\mathds{R}^d \rightarrow[0;1] $ for $\iota=1,\ldots,\mathcal{K}$,
is called a \emph{vectorial} \emph{generalized} \emph{downcrossing}
of $\ddot{\boldsymbol{\psi}}$ if for all $\mathbf{y}\in\mathds
{R}^d$ and for any $\iota=1,\ldots,\mathcal{K}$
\[
\mbox{for all } x_{\iota}<t_{\iota}(\mathbf{y}), \qquad\ddot{\psi
}_{\iota}(\mathbf{x};\mathbf{y}) \geq t_{\iota}(\mathbf{y}) \quad
\mbox {and}\quad\mbox{for all } x_{\iota}>t_{\iota}(\mathbf{y}),
\qquad \ddot {\psi}_{\iota}(\mathbf{x};\mathbf{y}) \leq t_{\iota}(
\mathbf{y}).
\]
\end{definition}

Clearly, if the function $\ddot{\psi}_{\iota}(\mathbf{x};\mathbf
{y})$ is decreasing in $\mathbf{x}$ (i.e., componentwise) for any
$\iota$,
then the vectorial generalized downcrossing $\mathbf{t}(\mathbf
{y})$ is unique, with $\mathbf{t}(\mathbf{y})\in(0;1)^\mathcal{K}$
for any $\mathbf{y}\in\mathds{R}^d$; furthermore $\ddot{\boldsymbol
{\psi}}(\mathbf{t}(\mathbf{y});\mathbf{y})=\mathbf{t}(\mathbf{y})$, provided that the solution exists.
Moreover, note that if $\ddot{\psi}_{\iota}(\mathbf{x};\mathbf
{y})=\ddot{\psi}_{\iota}(x_{\iota};\mathbf{y})$ for any $\iota
=1,\ldots,\mathcal{K}$, then each
component $t_{\iota}(\mathbf{y})$ of $\mathbf{t}(\mathbf{y})$
is simply the single generalized downcrossing of $\ddot{\psi}_{\iota
}(x_\iota;\mathbf{y})$, which
can be found by solving the equation $\ddot{\psi}_{\iota}(x;\mathbf
{y})=x$ (if the solution exists).

\begin{them}\label{thm4}
At each step $n$, suppose that for any given stratum $(t_j,w_l)$ the
allocation function
$\varphi_{jl} (\boldsymbol{\uppi}_n ;\widehat{\boldsymbol
{\gamma}}_{n},\mathbf{S}_{n}  )$ is decreasing in
$\boldsymbol{\uppi}_n$ (componentwise).
If the unique vectorial generalized downcrossing $\mathbf{t}
(\widehat{\boldsymbol{\gamma}}_{n},\mathbf{S}_{n}  )=
[t_{jl}(\widehat{\boldsymbol{\gamma}}_{n},\mathbf{S}_{n}) \dvt j
=0,\ldots,J; l = 0, \ldots,L]$
is a continuous function and $\boldsymbol{\varphi}(\mathbf
{t} (\boldsymbol{\gamma},\boldsymbol{\varsigma}
);\boldsymbol{\gamma},\boldsymbol{\varsigma})=\mathbf{t}
(\boldsymbol{\gamma},\boldsymbol{\varsigma}  )$, then
\[
\lim_{n\rightarrow\infty} \boldsymbol{\uppi}_n=\mathbf{t} (
\boldsymbol{\gamma},\boldsymbol{\varsigma} ) \quad\mbox {and}\quad \lim
_{n\rightarrow\infty} \uppi_n=E_{\mathbf{Z} }\bigl[\mathbf
{t} (\boldsymbol{\gamma},\boldsymbol{\varsigma} )\bigr]=\sum
_{j=0}^J\sum_{l=0}^Lt_{jl}(
\boldsymbol{\gamma},\boldsymbol {\varsigma}) p_{jl} \qquad\mbox{a.s.}
\]
\end{them}

\begin{pf}
See Appendix~\ref{A3}.
\end{pf}

%
\begin{ex}\label{exRDBCD}
The Reinforced Doubly-adaptive Biased Coin Design (RDBCD) is a class of
CARA procedures recently introduced by
Baldi Antognini and Zagoraiou
\cite{Baz12} in the case of
categorical covariates intended to target any desired allocation proportion
\[
\boldsymbol{\uppi}^{\ast}(\boldsymbol{\gamma})=\bigl[
\uppi^{\ast
}(j,l)\dvt j=0,\ldots, J; l=0,\ldots, L\bigr]\dvtx\Omega
\rightarrow (0,1)^{(J+1)\times(L+1)},
\]
which is a continuous function of the unknown model parameters.
Starting with a pilot stage performed to
derive an initial parameter estimation, at each step $n\geq2m$ let
$\widehat{\uppi}_{n}^{\ast}(j,l)$ be the estimate of the target within
stratum $(t_j,w_l)$ obtained using all the collected data up to that
step and $\widehat{p}_{jln}=n^{-1}N_n(j,l)$ the estimate of $p_{jl}$;
when the next patient
with covariate $\mathbf{Z}_{n+1}=(t_j,w_l)$ is ready to be
randomized, the RDBCD assigns him/her
to $A$ with probability
\[
\Pr\bigl(\delta_{n+1}=1\vert\Im_n, \mathbf{Z}_{n+1}=(t_j,w_l)
\bigr)=\varphi _{jl} \bigl(\uppi_{n}(j,l);\widehat{\uppi
}_{n}^{\ast}(j,l),\widehat{p}_{jln} \bigr),
\]
where the function $\varphi_{jl}(x;y,z)\dvtx(0,1)^3 \rightarrow[0,1]$
satisfies the following conditions:
\begin{enumerate}[(iii)]
\item[(i)] $\varphi_{jl}$ is decreasing in $x$ and increasing in $y$,
for any $ z \in(0,1)$;
\item[(ii)] $\varphi_{jl}(x;x,z)=x$ for any $z \in(0,1)$;
\item[(iii)] $\varphi_{jl}$ is decreasing in $z$ if $x<y$, and
increasing in $z$ if $x>y$;
\item[(iv)] $\varphi_{jl}(x;y,z)=1-\varphi_{jl}(1-x;1-y,z)$ for any $
z \in(0,1)$.
\end{enumerate}
First, observe that for the RDBCD (\ref{stratrand}) holds and thus,
from (i) and (ii), at each stratum $(t_j,w_l)$ the only generalized
downcrossing of $\varphi_{jl}$ is simply given by $\widehat{\uppi
}_{n}^{\ast}(j,l)$. Therefore, by Theorem~\ref{thm4}, $\lim_{n\rightarrow\infty} \uppi_n(j,l) ={\uppi}^*(j,l)$ a.s. for any
$j=0,\ldots, J$ and $l=0,\ldots, L$, due to the continuity of the
target, that is, $\lim_{n\rightarrow\infty} \boldsymbol{\uppi
}_n=\boldsymbol{\uppi}^*(\boldsymbol{\gamma})$ a.s.
\end{ex}

\subsection{Covariate-Adaptive designs with categorical covariates}

Theorem~\ref{thm4} can be naturally applied to CA procedures in the
case of categorical covariates by assuming, instead of (\ref
{randcaracat}), the following class of allocation rules:
%
%
\begin{equation}
\label{randcaracat2} \Pr(\delta_{n+1}=1\vert\Im_n,
\mathbf{Z}_{n+1}=\mathbf{z}_{n+1})=\varphi_{jl}
(\boldsymbol{\uppi}_n ;\mathbf{S}_{n} ),
\end{equation}
where now $\Im_n=\sigma(\delta_1,\ldots,\delta_n;\mathbf{Z}_{ {1}},
\ldots,\mathbf{Z}_{ {n}})$.
Moreover, from now on we let $\mathbf{t}^{B}= [1/2 \dvt j
=0,\ldots
,J; l = 0, \ldots,L]$.

\begin{ex}\label{exCABCD}
The Covariate-Adaptive Biased Coin Design (C-ABCD) \cite{Baz11} is a
class of stratified randomization procedures intended to achieve joint
balance. For any stratum $(t_j,w_l)$, let $F_{jl}(\cdot)\dvtx \mathds
{R}\rightarrow[0,1]$ be a non-increasing and symmetric function with
$F_{jl}(-x)=1-F_{jl}(x)$; the C-ABCD assigns the $(n+1)$st patient with
profile $\mathbf{Z}_{n+1}=(t_{j},w_{l})$ to $A$ with probability
%
%
\begin{equation}
\label{CABCD} \Pr \bigl(\delta_{n+1}=1\vert\Im_n,
\mathbf{Z}_{n+1}=(t_j,w_l) \bigr)
=F_{jl}\bigl[D_n(j,l)\bigr],
\end{equation}
where $D_n(j,l)=N_n(j,l) [2\uppi_n(j,l)-1 ]$ is the imbalance
between the two groups after $n$ steps within stratum $(t_j,w_l)$.
As showed in Remark~\ref{remABCD} and Example~\ref{exABCD} in the
case of {AA} procedures, Theorem~\ref{thm4}
still holds even if we assume different randomization functions at each step,
provided that the unique vectorial generalized downcrossing is the same
for any $n$. Indeed, it is trivial to see that rule (\ref{CABCD})
corresponds to
\[
\varphi_{jln} (\boldsymbol{\uppi}_n ;
\mathbf{S}_{n} )=\varphi_{jln} \bigl(
\uppi_n(j,l) ;\mathbf{S}_{n} \bigr)=F_{jl}
\bigl\{n \bigl[2\uppi_n(j,l)-1 \bigr]\widehat{p}_{jln}
\bigr\},
\]
and, from the properties of $F_{jl}$, $\varphi_{jln}$'s have $1/2$ as
unique downcrossing for any $n$; thus $\lim_{n\rightarrow\infty}
\boldsymbol{\uppi}_n=\mathbf{t}^{B}$, which clearly implies
marginal balance.

Moreover, when the covariate distribution is known
Baldi Antognini and Zagoraiou
\cite{Baz11}
suggested the following class of randomization rules:
\[
F_{jl}^q(x)= \bigl\{x^{q(p_{jl})}+1\bigr
\}^{-1},\qquad x\geq1,
\]
where $q(\cdot)$ is a decreasing function with $\lim_{t\rightarrow
0^+} q(t)=\infty$. Clearly, the above mentioned arguments and Theorem~\ref{thm4} guarantee the convergence to balance even if the covariate
distribution is unknown, by replacing at each step $p_{jl}$ with its
current estimate.
\end{ex}

Examples~\ref{exRDBCD} and~\ref{exCABCD} deal with procedures such
that, at every step $n$, the allocation rule $\varphi_{jl}$ depends
only on the current allocation proportion $\uppi_n(j,l)$, namely
satisfying (\ref{stratrand}). We now present additional examples where
$\varphi_{jl}$ is a function of the whole vectorial allocation
proportion~$\boldsymbol{\uppi}_n$.

%
\begin{ex}
Minimization methods \cite{Poc75,Tav74}
are stratified randomization procedures intended to achieve the
so-called marginal balance among covariates. In general,
they depend on the definition of a measure of overall imbalance among the
assignments which summarizes the imbalances between the treatment
groups for each level of every factor. Assuming the well-known variance
method proposed by Pocock and Simon \cite{Poc75},
the $(n+1)$st subject with
covariate profile $\mathbf{Z}_{n+1}=(t_{j},w_{l})$ is assigned to
treatment $A$ with probability
%
%
\begin{equation}
\label{peS1} \Pr\bigl(\delta_{n+1}=1\vert\Im_{n},
\mathbf{Z}_{n+1}=(t_{j},w_{l})\bigr)= %
\cases{ p, & \quad$D_{n}(t_{j}) + D_{n}(w_{l})<0$,\vspace*{1pt}
\cr
\frac{1}{2}, & \quad$D_{n}(t_{j}) +
D_{n}(w_{l})=0$,\vspace*{1pt}
\cr
1-p, & \quad$D_{n}(t_{j})
+ D_{n}(w_{l})>0$, } %
\end{equation}
where $p\in[1/2;1]$, $D_n(t_j)$ is the imbalance between the
two arms within the level $t_j$
of $T$ and, similarly, $D_n(w_l)$ represents the imbalance at the
category $w_l$ of $W$.
At each step $n$, note that $\sgn\{ D_{n}(t_{j})\}=\sgn\{n^{-1}
D_{n}(t_{j})\}$ where
%
%
\begin{equation}
\label{joint} n^{-1}D_{n}(t_{j})=\sum
_{l=0}^L \bigl[2 \uppi_n(j,l) -1
\bigr]\hat {p}_{jln}, \qquad\mbox{for any } j=0,\ldots,J
\end{equation}
and analogously for $D_{n}(w_{l})$. Thus, allocation rule (\ref{peS1})
corresponds to
\[
\varphi^{\mathrm{PS}}_{jl} (\boldsymbol{
\uppi}_n ;\mathbf{S}_{n} )= %
\cases{ p, &
\quad$\displaystyle\sum_{l=0}^{L} \biggl[
\uppi_n(j,l) - \frac
{1}{2} \biggr]\hat {p}_{jln} +
\sum_{j=0}^{J} \biggl[
\uppi_n(j,l) - \frac{1}{2} \biggr] \hat{p}_{jln}
<0$,
\cr
\displaystyle\frac{1}{2}, & \quad$\displaystyle\sum
_{l=0}^{L} \biggl[ \uppi_n(j,l) -
\frac{1}{2} \biggr] \hat{p}_{jln} + \sum
_{j=0}^{J} \biggl[ \uppi_n(j,l) -
\frac
{1}{2} \biggr] \hat{p}_{jln} =0$,
\cr
1-p,& \quad$
\displaystyle\sum_{l=0}^{L} \biggl[
\uppi_n(j,l) - \frac{1}{2} \biggr] \hat {p}_{jln} +
\sum_{j=0}^{J} \biggl[
\uppi_n(j,l) - \frac{1}{2} \biggr] \hat{p}_{jln}
>0$, } %
\]
and therefore the problem consists in finding the vectorial generalized
downcrossing of
$\boldsymbol{\varphi}^{\mathrm{PS}} (\boldsymbol{\uppi
}_n;\mathbf{S}_{n})=[\varphi^{\mathrm{PS}}_{jl}(\boldsymbol{\uppi}_n;
\mathbf{S}_{n})\dvt j =0,\ldots,J; l = 0, \ldots,L]$. Since at each step $n$,
$\varphi^{\mathrm{PS}}_{jl}(\boldsymbol{\uppi}_n; \mathbf{S}_{n})$ is
decreasing in $\uppi_{n}(j,l)$ for any $j =0,\ldots,J$ and
$l = 0, \ldots,L$, then the vectorial generalized downcrossing is
unique. It is
straightforward to see that
$\boldsymbol{\varphi}^{\mathrm{PS}}(\mathbf{t}^{B};\boldsymbol
{\varsigma
})=\mathbf{t}^{B}$ for every $n$ and thus
$\lim_{n\rightarrow\infty} \boldsymbol{\uppi}_n=\mathbf{t}^{B}$ a.s.
\end{ex}

%
\begin{ex}
In order to include minimization methods and stratified randomization
procedures in a unique framework, Hu and Hu \cite{Hu12}
have recently suggested to assign subject $(n+1)$ belonging to the
stratum $(t_{j},w_{l})$ to $A$ with probability
%
%
\begin{equation}
\label{huhuproc} \Pr\bigl(\delta_{n+1}=1\vert\Im_{n},
\mathbf{Z}_{n+1}=(t_{j},w_{l})\bigr)= %
\cases{ p, & \quad$\bar{D}_{n}(j,l)<0$,\vspace*{1pt}
\cr
\tfrac{1}{2}, &
\quad$\bar{D}_{n}(j,l)=0$,\vspace*{1pt}
\cr
1-p, & \quad$\bar{D}_{n}(j,l)>0$,
} %
\end{equation}
where the overall measure of imbalance
\[
\bar{D}_{n}(j,l)= \omega_{g}D_{n}+
\omega_{T}D_{n}(t_{j})+\omega
_{W}D_{n}(w_{l})+\omega_{s}D_{n}(j,l)
\]
is a weighted average of the three types of imbalances actually
observed (global, marginal and within-stratum), with non-negative
weights $\omega
_{g}$ (global), $\omega_{T}$ and $\omega_{W}$ (covariate marginal)
and $\omega_{s}$ (stratum) chosen such that $\omega_{g}+\omega
_{T}+\omega
_{W}+\omega_{s}=1$.

By choosing the weights $\omega_g$, $\omega_T$, $\omega_W$ such
that\vspace*{-1pt}
%
%
\begin{equation}
\label{condcin2} (JL+J+L)\omega_g+J\omega_W+L
\omega_T<1/2,
\end{equation}
the authors proved that the probabilistic structure of the within
stratum imbalance is that of a positive recurrent Markov chain and this
implies that procedure (\ref{huhuproc}) is asymptotically balanced,
both marginally and jointly. However, as stated by the authors, only
strictly positive choices of the stratum weight $\omega_{s}$ satisfy
(\ref{condcin2}), and thus their result cannot be applied to Pocock
and Simon's minimization method.

The asymptotic behaviour of Hu and Hu's design can be illustrated in a
different way by applying Theorem~\ref{thm4}. Since $\sgn\{\bar
{D}_{n}(j,l) \}=\sgn\{n^{-1}\bar{D}_{n}(j,l) \}$ and\vspace*{-1pt}
%
%
\begin{equation}
\label{marginal} n^{-1}D_n=2\uppi_n-1=\sum
_{j=0}^J\sum_{l=0}^L
\bigl[2\uppi _n(j,l)-1 \bigr] \hat{p}_{jln},
\end{equation}
from (\ref{joint}) it follows that\vspace*{-1pt}
\begin{eqnarray*}
  \sgn\bigl\{n^{-1}\bar{D}_{n}(j,l)
\bigr\}
 &=& \sgn\Biggl\{\omega_{g}\sum_{j=0}^J
\sum_{l=0}^L \biggl[\uppi_n(j,l)-
\frac{1}{2} \biggr] \hat {p}_{jln}+ \omega_{T}\sum
_{l=0}^L \biggl[\uppi_n(j,l)-
\frac
{1}{2} \biggr] \hat{p}_{jln}
\\[-1pt]
&&\phantom{\sgn\Biggl\{}{}+ \omega_{W} \sum_{j=0}^J
\biggl[\uppi_n(j,l)- \frac{1}{2} \biggr] \hat
{p}_{jln}+ \omega_{s} \biggl[\uppi_n(j,l)-
\frac{1}{2} \biggr] \hat{p}_{jln} \Biggr\}. 
\end{eqnarray*}
Thus, at each step $n$ procedure (\ref{huhuproc}) corresponds to an
allocation rule $\varphi^{\mathrm{HH}}_{jl} (\boldsymbol{\uppi}_n
;\mathbf{S}_{n}  )$
which is decreasing in $\uppi_{n}(j,l)$ for any $j =0,\ldots,J$ and
$l = 0, \ldots,L$. Since $\boldsymbol{\varphi}^{\mathrm
{HH}}(\mathbf{t}^{B};\boldsymbol{\varsigma})=\mathbf{t}^{B}$,
then the unique vectorial generalized downcrossing is $\mathbf{t}^{B}$ for any $n$ and therefore $\lim_{n\rightarrow\infty}
\boldsymbol{\uppi}_n=\mathbf{t}^{B}$ a.s.\vspace*{-1pt}
\end{ex}

Under the same arguments, it can be easily proved the convergence to
balance of several extensions of minimization methods (see, e.g.,
\cite{Her05,Signor93}), since at each step $n$ every type of imbalance
(global, marginal and within-stratum) is a linear combination of the
allocation proportions $\uppi_n(j,l)$'s.\vspace*{-1pt}

%
\begin{ex}\label{exatk}
Assuming the liner homoscedastic model without treatment/covariate
interaction in the form\vspace*{-1pt}
%
%
\begin{equation}
\label{linearsenza} E(Y_{i}) =\delta_{i} \mu_{A}+(1-
\delta_{i}) \mu_{B}+ \widetilde {f}(\mathbf{z}_{i})^{t}
\boldsymbol{\beta},\qquad i\geq1,
\end{equation}
where $\widetilde{f}(\cdot)$ is a known vector function and
$\boldsymbol{\beta}$ is a vector of common regression parameters.

Put $\mathcal{F}_n= [\widetilde{f}(\mathbf{z}_{i})^{t} ]$ and $\mathds{F}_n=[\mathbf{1}_n \dvtx \mathcal{F}_n]$,
Atkinson \cite{Atk82} introduced his biased coin design by assigning
the $(n+1)$st patient to $A$ with probability\vspace*{-1pt}
%
%
\begin{eqnarray}
\label{atkinsondes} %
&& \Pr(\delta_{n+1}=1\vert
\Im_n, \mathbf{Z}_{n+1})
\nonumber
\\[-8.5pt]
\\[-8.5pt]
&&\quad=\frac{\{1-(1; \widetilde{f} (\mathbf{z}_{n+1})^t )
(\mathds{F}_n^t \mathds{F}_n)^{-1} \mathbf{b}_n \}^2}{\{1-(1;
\widetilde{f} (\mathbf{z}_{n+1})^t )
(\mathds{F}_n^t \mathds{F}_n)^{-1} \mathbf{b}_n \}^2+\{1+(1;
\widetilde{f} (\mathbf{z}_{n+1})^t )
( \mathds{F}_n^t \mathds{F}_n)^{-1} \mathbf{b}_n \}^2},
\nonumber
\end{eqnarray}
where $\mathbf{b}_n^t=(2\delta_1-1, \ldots,2\delta_n-1)\mathds
{F}_n$ is usually called the imbalance vector.

As showed in \cite{Baz11}, in the presence of all interactions among
covariates we obtain
\[
\mathbf{b}_n^t=\bigl(D_n,D_n(t_1),
\ldots,D_n(t_J), D_n(w_1),
\ldots ,D_n(w_L),D_n(1,1),
\ldots,D_n(J,L)\bigr)
\]
and Atkinson's procedure (\ref{atkinsondes}) becomes a stratified
randomization rule with
%
%
\begin{eqnarray}
\label{DABCD} &&\Pr\bigl(\delta_{n+1}=1\vert\Im_n,
\mathbf{Z}_{n+1}=(t_j,w_l)\bigr)
\nonumber
\\[-8pt]
\\[-8pt]
&&\quad=
\frac
{  ( 1- {D_n(j,l)}/{N_n(j,l)} )^2 }{ ( 1-
{D_n(j,l)}/{N_n(j,l)} )^2+ ( 1+ {D_n(j,l)}/{N_n(j,l)} )^2}.\nonumber
\end{eqnarray}
Clearly, procedure (\ref{DABCD}) corresponds to
\[
\varphi_{jl} (\boldsymbol{\uppi}_n ;
\mathbf{S}_{n} )=\frac{  [ 1- \uppi_n(j,l) ]^2 }{ [ 1- \uppi
_n(j,l) ]^2+\uppi_n(j,l)^2},
\]
so (\ref{stratrand}) holds; thus, by Theorem~\ref{thm4}, $\lim_{n\rightarrow\infty} \boldsymbol{\uppi}_n=\mathbf{t}^{B}$.

When the model is not full, then $\mathbf{b}_n$ contains
all the imbalance terms corresponding to the included
interactions. Thus, from (\ref{joint}) and (\ref{marginal}),
$(1;\widetilde{f} (\mathbf{z}_{n+1})^t )
(\mathds{F}_n^t \mathds{F}_n)^{-1} \mathbf{b}_n$ is a linear
function of the allocation proportion $\boldsymbol{\uppi}_n$, so that
Theorem~\ref{thm4} can be applied by the previous arguments.
\end{ex}

\section{Downcrossing and stochastic approximation methods}\label{sec7}

By combining the concept of downcrossing and stopping times of
stochastic processes, we demonstrated the almost
sure convergence of the treatment allocation proportion for a vast
class of adaptive procedures. In general, this is due to the fact that
the asymptotic behavior of $\uppi_n$ coincides with that of the sequence
of downcrossing points of the corresponding allocation function. An
alternative way to characterize the same large-sample behavior is
provided by the Stochastic Approximation (SA) methods (see, e.g., \cite
{Ben90,Duf97,Kus03,Lai03}) and the asymptotic theory of
super-martingales \cite{Hall80}. Indeed, considered now the AA
procedures of Section~\ref{sec3}, from (\ref{SA}) at each step $n\geq1$,
%
%
\begin{eqnarray}
\label{SA1} %
\uppi_{n+1} &=&
\uppi_{n} \biggl(\frac{n}{n+1} \biggr) + \frac
{1}{n+1} \bigl
\{\Delta M_{n+1} + \varphi^{\mathrm{AA}} (\uppi _n ) \bigr
\}
\nonumber
\\[-8pt]
\\[-8pt]
&=& \uppi_{n} - \frac{1}{n+1} \bigl\{ \uppi_n -
\varphi^{\mathrm
{AA}} (\uppi _n ) \bigr\}+ \frac{\Delta M_{n+1}}{n+1}.
\nonumber
\end{eqnarray}
Therefore, the allocation proportion follows the classical
Robbins--Monro \cite{Robb51} recursive relation:
%
%
\begin{equation}
\label{SA2} \uppi_{n+1} = \uppi_{n} - a_nH(
\uppi_n) +a_n\Delta M_{n+1},
\end{equation}
where $H( x)=x-\varphi^{\mathrm{AA}}  (x  )$ and $a_n=(n+1)^{-1}$.
Note that:
\begin{itemize}
\item$\{\Delta M_{n}\}$ is a sequence of bounded martingale
differences, so that for any $n$
\[
E[\Delta M_{n+1}\vert\Im_n]=0 \quad\mbox{and}\quad E\bigl[
\Delta M_{n+1}^2\vert\Im _n\bigr] =
\varphi^{\mathrm{AA}} (\uppi_n )\bigl[1-\varphi ^{\mathrm{AA}} (
\uppi_n )\bigr]\leq1 ;
\]
\item$\lim_{n\rightarrow\infty}a_n=0$, $\sum_{i=1}^{\infty
}a_i=\infty$ and $\sum_{i=1}^{\infty}a_i^2<\infty$;
\item the function $H(\cdot)$ is increasing, since $\varphi
^{\mathrm{AA}}(\cdot)$ is assumed to be decreasing, and furthermore
$(x-t)H(x)>0$ since $t$ is the unique downcrossing of $\varphi
^{\mathrm{AA}}(\cdot)$.
\end{itemize}
Therefore, it follows that $\uppi_n\rightarrow t$ a.s.

The same asymptotic result can be obtained via a super-martingale
approach since, from (\ref{SA1})
\[
E\bigl[ (\uppi_{n+1}-t)^2 \vert\Im_n \bigr] =
(\uppi_{n}-t)^2- 2a_n(\uppi_{n}-t)
H(\uppi_n) +a^2_n \bigl\{H^2(
\uppi_n)+E\bigl[\Delta M_{n+1}^2\vert\Im
_n\bigr] \bigr\},
\]
where the last term of the r.h.s. is asymptotically negligible, due to
the properties of $a_n$, $H(\cdot)$ and $\Delta M_{n+1}$, and $(\uppi
_{n}-t) H(\uppi_n)$ is always non-negative. Thus, the quantity
%
%
\begin{equation}
\label{SA4} (\uppi_n-t)^2 \mbox{ is a non-negative
almost super-martingale},
\end{equation}
namely it is asymptotically equivalent to a non-negative
super-martingale; therefore it converges almost surely and, in our
setting, it vanishes asymptotically. If we further assume $\varphi
^{\mathrm{AA}}$ differentiable, then
\[
\frac{\partial\varphi^{\mathrm{AA}}(x)}{\partial x} \bigg|_{x=t}=\varphi'(t)<0,
\]
so that from Fabian's theorem \cite{Fab68}
%
%
\begin{equation}
\label{SA5} \sqrt{n}(\uppi_n-t) \hookrightarrow N \biggl(0;
\frac
{t(1-t)}{1-2\varphi'(t)} \biggr),
\end{equation}
since $\lim_{n\rightarrow\infty}E[\Delta M_{n+1}^2\vert\Im_n]
=\varphi^{\mathrm{AA}}  (t  )[1-\varphi^{\mathrm{AA}}
(t
)]=t(1-t)$ (due to the continuity of $\varphi^{\mathrm{AA}}$).
The asymptotic variance in (\ref{SA5}) could help distinguish between
different AA rules intended to achieve the same target allocation
proportion; clearly, this variance increases as $\varphi'(t)$ grows
(i.e. as the random component in the assignments increases).

\begin{ex}
Adopting Wei's Adaptive BCD in (\ref{weiallfunc}) with $\mathfrak
{f}(\cdot)$ differentiable, then
%
%
\begin{equation}
\label{asnormwei} \lim_{n\rightarrow\infty}\uppi_n=
\frac{1}{2} \qquad\mbox{a.s.} \quad\mbox{and} \quad\sqrt {n} \biggl(
\uppi_n-\frac{1}{2} \biggr) \hookrightarrow N \biggl(0;
\frac
{1}{4 [1-2\mathfrak{f}' ( {1}/{2} )
]} \biggr).
\end{equation}
While assuming CR design
\[
\lim_{n\rightarrow\infty}\uppi_n=\frac{1}{2} \qquad
\mbox{a.s.} \quad\mbox{and}\quad\sqrt {n} \biggl(\uppi_n-
\frac{1}{2} \biggr) \hookrightarrow N \biggl(0; \frac
{1}{4} \biggr),
\]
namely under CR the asymptotic variance of the allocation proportion is
always greater than Wei's one (since $\mathfrak{f}$ is decreasing).
This reduction in terms of asymptotic variance lies in the fact that
Wei's rule favors at each step the assignments of the under-represented
treatment, that is, it is adapted to the sequence of previous allocations.
\end{ex}

\begin{rem}
Even if SA theory can be applied in the context of adaptive procedures,
we would like to stress some differences between them:
\begin{itemize}
\item in the classical Robbins--Monro scheme, there is a controllable
design variable taking values in $\mathbb{R}$ that follows itself the
SA recursion, while in our setting the design space is discrete, since
$\delta_n \in\{0;1\}$, and only the allocation proportions $\uppi_n$s
follow (\ref{SA2});
\item within SA framework the function $H(\cdot)$ is analytically
unknown and it cannot be observed directly, but it could be observed
only with a stochastic perturbation; whereas in our setting the
allocation function is the only ingredient chosen by the experimenter
and thus it is completely known (while the assignments $\delta_n$s are
randomly generated by the allocation rule).
\end{itemize}
\end{rem}

As regards the other types of adaptive procedures, (\ref{SA4}) is
still a non-negative almost super-martingale provided that the
downcrossing $t$ of $\varphi^{\mathrm{AA}}(\cdot)$ is substituted by the
generalized (vectorial) downcrossing of the corresponding allocation
function. For instance, in the RA case $t\mapsto t (\widehat
{\boldsymbol{\gamma}}_{n}  )$ and the asymptotic behavior of
$\uppi_n$ coincides with that of $t (\widehat{\boldsymbol
{\gamma
}}_{n} )$; therefore, assuming $t (\cdot )$
continuous, as $n$ grows $\uppi_n\rightarrow t (\boldsymbol
{\gamma
} )$ a.s. Furthermore, by adding suitable continuity and
differentiability conditions for $t (\cdot )$ and the
allocation function, it is possible to derive the asymptotic normality
of the allocation proportions as in (\ref{asnormwei}) (see, e.g.,
\cite{Hu04} for RA procedures and \cite{Zha07} for CARA designs).
Note that the case of CARA designs with categorical covariates is a
multidimensional SA scheme where, at each step $n$, (i) the evolution at
each stratum depends, in general, on the information gathered up to
that step from all the strata and (ii) the constant $a_n$ should be
replaced by the random vector $\mathbf{a}_n=
[N_{n+1}(j,l)^{-1}\mathbh{1}_{\{Z_i=(t_{j},w_{l})\}}, j=0,\ldots
,J;l=0,\ldots,L  ]$ and therefore the Robbins--Siegmund's lemma
(1971) should be applied (see \cite{Lai03,Rob71}).

\begin{appendix}
\section*{Appendix}\label{app}
\subsection{Proof of Theorem \texorpdfstring{\protect\ref{thm2}}{4.1}}\label{A1}

At each step $n$, consider the squared integrable martingale process $\{
M_{n}; \Im_n \}$, where $M_{n}=\sum_{i=1}^{n}\Delta M_i=\sum_{i=1}^{n} \{\delta_i-E(\delta_i|\Im_{i-1}) \}$ and $\Im
_n=\sigma(\delta_1,\ldots,\delta_n;Y_1,\ldots,Y_n)$.

Let $\lambda_n=\max \{s\dvt 2m+1 \leq s \leq n, \uppi_s \leq
t(\widehat{\boldsymbol{\gamma}}_{s})  \}$, with $\max
\emptyset=2m$. Thus at each step $i>\lambda_n$, $\varphi^{\mathrm
{RA}}
(\uppi_{i} ;\widehat{\boldsymbol{\gamma}}_{i}  ) \leq
t(\widehat{\boldsymbol{\gamma}}_{i})$ and therefore
\begin{eqnarray*}
\widetilde{N}_n& = &\widetilde{N}_{\lambda_n+1} +
\sum_{k=\lambda
_n+2}^{n} \Delta M_k+
\sum_{k=\lambda_n+2}^{n} \varphi^{\mathrm
{RA}} (
\uppi_{k-1} ;\widehat{\boldsymbol{\gamma}}_{k-1} )
\\
&\leq& \widetilde{N}_{\lambda_n} +1+M_n -M_{\lambda_n+1} +
\sum_{k=\lambda_n+2}^{n} t (\widehat{\boldsymbol{
\gamma}}_{k-1} ). %
\end{eqnarray*}
Since $\widetilde{N}_{\lambda_n} \leq\lambda_n t (\widehat
{\boldsymbol{\gamma}}_{\lambda_n}  )$ we obtain
\begin{eqnarray*}
\widetilde{N}_n - n t (\widehat{\boldsymbol{
\gamma}}_{n} )&\leq& \Biggl(\lambda_n t (\widehat{
\boldsymbol{\gamma }}_{\lambda_n} ) - \sum_{k=2}^{\lambda_n+1}
t (\widehat {\boldsymbol{\gamma}}_{k-1} ) \Biggr)+ M_n
-M_{\lambda_n+1} +1-t (\widehat{\boldsymbol{\gamma}}_{0} )
\\
&&{}- \Biggl( n t (\widehat{\boldsymbol{\gamma}}_{n} )- \sum
_{k=1}^{n} t (\widehat{\boldsymbol{
\gamma}}_{k-1} ) \Biggr), %
\end{eqnarray*}
where $t (\widehat{\boldsymbol{\gamma}}_{0} )=t_0\in
[0;1]$ is a constant depending on the initial stage. Furthermore, as $n
\rightarrow\infty$, at least one of the number of assignments to the
treatments, $\widetilde{N}_n$ and $(n-\widetilde{N}_n)$, tends to
infinity a.s. As showed in \cite{Hu09}, in either case $\widehat
{\boldsymbol{\gamma}}_{n}$
has finite limit so that, from the properties of $t (\widehat
{\boldsymbol{\gamma}}_{n}  )$, almost surely there exists a $v$
such that
\setcounter{equation}{0}
%
\begin{equation}
\label{l11} t (\widehat{\boldsymbol{\gamma}}_{n} )\rightarrow v
\qquad\mbox{a.s.}
\end{equation}
and so $\lim_{n\rightarrow\infty}t (\widehat{\boldsymbol
{\gamma}}_{n}  )- n^{-1}\sum_{k=1}^{n} t (\widehat
{\boldsymbol{\gamma}}_{k-1}  ) = 0$ a.s.
As $n\rightarrow\infty$, then $\lambda_n\rightarrow\infty$ or
$\sup_n \lambda_n< \infty$; in either case,
$\lim_{n\rightarrow\infty} n^{-1}\lambda_n [t (\widehat
{\boldsymbol{\gamma}}_{\lambda_n}  )- \lambda_n^{-1} \sum_{k=1}^{\lambda_n} t (\widehat{\boldsymbol{\gamma}}_{k}
)  ]=0$ a.s.
and therefore
%
%
\begin{equation}
\label{lim1} \bigl[\uppi_n - t(\widehat{\boldsymbol{
\gamma}}_{n})\bigr] ^+ \rightarrow0 \qquad\mbox{a.s.}
\end{equation}
Analogously,
%
%
\begin{equation}
\label{lim2} \bigl[(1-\uppi_n) - \bigl(1- t(\widehat{\boldsymbol{
\gamma }}_{n}) \bigr) \bigr]^+ \rightarrow0 \qquad\mbox{a.s.}
\end{equation}
From (\ref{lim1}) and (\ref{lim2}), as $n$ tends to infinity $\uppi_n
- t(\widehat{\boldsymbol{\gamma}}_{n}) \rightarrow0 $ a.s. and
by (\ref{l11})
$\lim_{n\rightarrow\infty} {\uppi_{n}}=\lim_{n\rightarrow\infty}
t(\widehat{\boldsymbol{\gamma}}_{n}) =v $ a.s.
Since $0< v <1$, then $0< 1-v <1$ and thus $\lim_{n\rightarrow\infty
} \widetilde{N}_n \rightarrow\infty$ a.s. and $ \lim_{n\rightarrow\infty}(n-\widetilde{N}_n) \rightarrow\infty$ a.s.
Therefore, $\lim_{n\rightarrow\infty}\widehat{\boldsymbol{\gamma
}}_{n}\rightarrow\boldsymbol{\gamma}$ a.s. and from the continuity
of the downcrossing $\lim_{n\rightarrow\infty} t(\widehat
{\boldsymbol{\gamma}}_{n})= t(\boldsymbol{\gamma})=v$ a.s., that
is, $\lim_{n\rightarrow\infty} \uppi_n= t(\boldsymbol{\gamma}) $ a.s.

\subsection{Proof of Theorem \texorpdfstring{\protect\ref{thmCARAc}}{5.1}}\label{A2}

If $\varphi^{\mathrm{CARA}}$ is decreasing in $\uppi_n$, then
$\tilde{\varphi
}_{\mathbf{Z}} $ is also decreasing in $\uppi_n$, so that the
generalized downcrossing
is unique and lies in $(0;1)$. Letting now $\Im_n=\sigma(\delta
_1,\ldots,\delta_n;Y_1,\ldots,Y_n;\mathbf{Z}_{1},\ldots
,\break \mathbf{Z}_{n} )$, then $E(\delta_i|\Im_{i-1})=E_{\mathbf
{Z}_{i}} [ \varphi (\uppi_{i-1} ;\widehat{\boldsymbol
{\gamma}}_{i-1},\mathbf{S}_{i-1}, f(\mathbf{Z}_{i}) )
]$ and $\Delta M_i=\delta_i-E(\delta_i|\Im_{i-1})$.
Then $\{\Delta M_i; i\geq1\} $ is a sequence of bounded martingale
differences with $|\Delta M_i |\leq1$ for any $i\geq1$; thus
$\{M_{n}=\sum_{i=1}^{n}\Delta M_i ; \Im_n \}$ is a martingale with
$\sum_{k=1}^{n} E[(\Delta M_i)^2 |\Im_{k-1}] \leq n$.
Let $\zeta_n=\max \{\vartheta\dvt 2m+1 \leq\vartheta\leq n,
\uppi
_\vartheta\leq\tilde{t}_{\mathbf{Z}} (\widehat{\boldsymbol
{\gamma}}_{\vartheta},\mathbf{S}_\vartheta)  \}$, with
$\max\emptyset=2m$. So that $\forall i>\zeta_n$ we have
$\tilde{\varphi}_{\mathbf{Z}}  (\uppi_{i} ;\widehat
{\boldsymbol{\gamma}}_{i},\mathbf{S}_i ) \leq\tilde
{t}_{\mathbf{Z}}
(\widehat{\boldsymbol{\gamma}}_{i},\mathbf{S}_i)$. Note that
\begin{eqnarray*}
\widetilde{N}_{n} &=& \widetilde{N}_{\zeta_n+1} +
\sum_{k=\zeta
_n+2}^{n} \Delta M_k+
\sum_{k=\zeta_n+2}^{n} E(\delta_{k}| \Im
_{k-1})
\\
&\leq& \widetilde{N}_{\zeta_n}+1 +M_n -M_{\zeta_n+1} +
\sum_{k=\zeta_n+2}^{n} \tilde{\varphi}_{\mathbf{Z}}
(\uppi _{k-1} ;\widehat{\boldsymbol{\gamma}}_{k-1},\mathbf
{S}_{k-1} )
\\
&<& \widetilde{N}_{\zeta_n} +1+M_n -M_{\zeta_n+1} + \sum
_{k=\zeta
_n+2}^{n} \tilde{t}_{\mathbf{Z}} (
\widehat{\boldsymbol {\gamma}}_{k-1},\mathbf{S}_{k-1} )
\\
&=& \widetilde{N}_{\zeta_n}+1+M_n -M_{\zeta_n+1} + \sum
_{k=1}^{n} \tilde{t}_{\mathbf{Z}} (
\widehat{\boldsymbol{\gamma }}_{k-1},\mathbf{S}_{k-1} ) -
\sum_{k=1}^{\zeta_n+1} \tilde{t}_{\mathbf{Z}} (
\widehat{\boldsymbol{\gamma }}_{k-1},\mathbf{S}_{k-1} ).
\end{eqnarray*}
Since $ \widetilde{N}_{\zeta_n} \leq\zeta_n \tilde{t}_{\mathbf
{Z}} (\widehat{\boldsymbol{\gamma}}_{\zeta_n},\mathbf
{S}_{\zeta_n})$, then
\begin{eqnarray*}
\widetilde{N}_{n} - n \tilde{t}_{\mathbf{Z}} (\widehat {
\boldsymbol{\gamma}}_{n},\mathbf{S}_{n}) &\leq& \Biggl(
\zeta_n \tilde{t}_{\mathbf{Z}} (\widehat{\boldsymbol {
\gamma}}_{\zeta_n},\mathbf{S}_{\zeta_n}) - \sum
_{k=2}^{\zeta
_n+1} \tilde{t}_{\mathbf{Z}} (\widehat{
\boldsymbol{\gamma }}_{k-1},\mathbf{S}_{k-1} ) \Biggr)
\\
&&{}+
M_n -M_{\zeta_n+1} +1-\tilde{t}_{\mathbf{Z}} (\widehat {
\boldsymbol{\gamma}}_{0},\mathbf{S}_{0} )- \Biggl( n
\tilde{t}_{\mathbf{Z}} (\widehat{\boldsymbol{\gamma }}_{n},
\mathbf{S}_{n})- \sum_{k=1}^{n}
\tilde{t}_{\mathbf
{Z}} (\widehat{\boldsymbol{\gamma}}_{k-1},
\mathbf{S}_{k-1} ) \Biggr).
\end{eqnarray*}
Moreover, as $n \rightarrow\infty$, at least one of
the number of assignments to the treatments, $\widetilde{N}_n$ and
$(n-\widetilde{N}_n)$, tends to infinity a.s. In either case from the
properties of $\tilde{t}_{\mathbf{Z}} (\widehat{\boldsymbol
{\gamma}}_{n},\mathbf{S}_{n} )$, almost surely there exists
a $\tilde{\upsilon}$ such that
%
%
\begin{equation}
\label{l1} \tilde{t}_{\mathbf{Z}} (\widehat{\boldsymbol{\gamma
}}_{n},\mathbf{S}_{n} ) \rightarrow\tilde{\upsilon}
\qquad\mbox{a.s.}
\end{equation}
and so
\[
\tilde{t}_{\mathbf{Z}} (\widehat{\boldsymbol{\gamma }}_{n},
\mathbf{S}_{n})- \frac{1}{n}\sum
_{k=1}^{n} \tilde {t}_{\mathbf{Z}} (\widehat{
\boldsymbol{\gamma }}_{k-1},\mathbf{S}_{k-1} ) \rightarrow0
\qquad\mbox{a.s.}
\]
As $n\rightarrow\infty$, then $\zeta_n\rightarrow\infty$ or $\sup_n \zeta_n< \infty$; in either case,
\[
\frac{\zeta_n}{n} \Biggl\{\tilde{t}_{\mathbf{Z}} (\widehat {\boldsymbol{
\gamma}}_{\zeta_n},\mathbf{S}_{\zeta_n}) - \frac
{1}{\zeta_n} \sum
_{k=1}^{\zeta_n} \tilde{t}_{\mathbf{Z}} (
\widehat{\boldsymbol{\gamma }}_{k}, \mathbf{S}_{k} )
\Biggr\}\rightarrow0 \qquad\mbox{a.s.}
\]
and therefore
%
%
\begin{equation}
\label{lim11} \bigl[\uppi_n - \tilde{t}_{\mathbf{Z}} (\widehat{
\boldsymbol {\gamma}}_{n}, \mathbf{S}_{n} )\bigr] ^+
\rightarrow0 \qquad \mbox{a.s.}
\end{equation}
Analogously,
%
%
\begin{equation}
\label{lim21} \bigl[(1-\uppi_n) - \bigl(1-\tilde{t}_{\mathbf{Z}}
(\widehat{\boldsymbol{\gamma}}_{n}, \mathbf{S}_{n} )
\bigr) \bigr]^+ \rightarrow0 \qquad\mbox{a.s.}
\end{equation}
From (\ref{lim11}) and (\ref{lim21}), $\lim_{n\rightarrow\infty
}\uppi_n - \tilde{t}_{\mathbf{Z}}  (\widehat{\boldsymbol
{\gamma}}_{n}, \mathbf{S}_{n} )=0$ a.s.
and therefore by (\ref{l1}) $\lim_{n\rightarrow\infty} {\uppi
_{n}}=\lim_{n\rightarrow\infty} \tilde{t}_{\mathbf{Z}}
(\widehat{\boldsymbol{\gamma}}_{n}, \mathbf{S}_{n} )
=\tilde{\upsilon}$ a.s.
Since $0< \tilde{\upsilon} <1$, then $0< 1-\tilde{\upsilon} <1$ and
$\lim_{n\rightarrow\infty} \widetilde{N}_n \rightarrow\infty$
a.s. and $ \lim_{n\rightarrow\infty}(n-\widetilde{N}_n)
\rightarrow\infty$ a.s.
Therefore, $\lim_{n\rightarrow\infty}\widehat{\boldsymbol{\gamma
}}_{n}\rightarrow\boldsymbol{\gamma}$ a.s. and also $\lim_{n\rightarrow\infty}\mathbf{S}_{n}\rightarrow\boldsymbol
{\varsigma}$ a.s., so that from the continuity of the downcrossing
$\lim_{n\rightarrow\infty} \tilde{t}_{\mathbf{Z}}
(\widehat{\boldsymbol{\gamma}}_{n}, \mathbf{S}_{n} )=
\tilde{t}_{\mathbf{Z}}(\boldsymbol{\gamma},\boldsymbol
{\varsigma})=\tilde{\upsilon}$ a.s., namely $\lim_{n\rightarrow
\infty} \uppi_n= \tilde{t}_{\mathbf{Z}}(\boldsymbol{\gamma
},\boldsymbol{\varsigma}) $ a.s.

\subsection{Proof of Theorem \texorpdfstring{\protect\ref{thm4}}{6.1}}\label{A3}

At each step $n$, let $M_{n}(j,l)=\sum_{i=1}^{n}\Delta M_i(j,l)=
\sum_{i=1}^{n}\{\delta_i-E(\delta_i| \mathfrak{G}_{i-1})\}
\mathbh{1}_{\{Z_i=(t_j,w_l)\}}$, where
$\mathfrak{G}_i=\sigma(\Im_i, \mathbf{Z}_{i+1})$.
Therefore, at each stratum $(t_j,w_l)$, $\{
\Delta M_i(j,l); i\geq1\} $ is a sequence of bounded martingale
differences with $|\Delta M_i (j,l) |\leq1$ for any $i\geq1$ and
thus, $\{M_{n}(j,l); \mathfrak{G}_n \}$ is a squared integrable
martingale with $\sum_{k=1}^{n} E[(\Delta M_i(j,l))^2 \vert\mathfrak
{G}_{k-1}] \leq n$.

Let $\xi_n(j,l)=\max \{s\dvt 2m+1 \leq i \leq n, \uppi_i(j,l)
\leq
t_{jl}(\widehat{\boldsymbol{\gamma}}_{i},\mathbf{S}_{i} )
\}$, with $\max\emptyset=2m$, then there exists a given
stratum $(t_{j'},w_{l'})$ such that $\xi_n(j',l')=\max_{jl} \xi
_n(j,l)$. Therefore, for any $i>\xi_n(j',l')$, at each stratum $\uppi
_i(j,l) > t_{jl}$ and, by Definition~\ref{DC3}, $\varphi_{jl}
(\boldsymbol{\uppi}_i ;\widehat{\boldsymbol{\gamma
}}_{i},\mathbf{S}_{i}  ) \leq t_{jl}(\widehat{\boldsymbol{\gamma
}}_{i},\mathbf{S}_{i})$. Thus
\begin{eqnarray*}
\widetilde{N}_n\bigl(j',l'
\bigr) &=&\widetilde{N}_{\xi_n(j',l')+1}\bigl(j',l'
\bigr) + \sum_{i=\xi_n(j',l')+2}^{n} \Delta
M_i\bigl(j',l'\bigr)+ \sum
_{i=\xi
_n(j',l')+2}^{n} E(\delta_{i}\vert
\mathfrak{G}_{i-1}) \mathbh{1}_{\{
Z_i=(t_{j'},w_{l'})\}}
\\
&\leq& \widetilde{N}_{\xi_n(j',l')}\bigl(j',l'
\bigr)+1 +M_n\bigl(j',l'\bigr)
-M_{\xi
_n(j',l')+1}\bigl(j',l'\bigr)
\\
&&{}+ \sum
_{i=\xi_n(j',l')+2}^{n} \varphi_{j'l'} (\boldsymbol{
\uppi}_{i-1} ;\widehat{\boldsymbol{\gamma }}_{i-1},
\mathbf{S}_{i-1} ) \mathbh{1}_{\{
Z_i=(t_{j'},w_{l'})\}}
\\
&<& \widetilde{N}_{\xi_n(j',l')}\bigl(j',l'\bigr)
+1+M_n\bigl(j',l'\bigr)
-M_{\xi
_n(j',l')+1}\bigl(j',l'\bigr)
\\
&&{}+ \sum
_{i=\xi_n(j',l')+2}^{n} t_{j'l'}(\widehat {\boldsymbol{
\gamma}}_{i-1},\mathbf{S}_{i-1}) \mathbh{1}_{\{
Z_i=(t_{j'},w_{l'})\}}
\\
&=& \widetilde{N}_{\xi_n(j',l')}\bigl(j',l'
\bigr)+1+M_n\bigl(j',l'\bigr)
-M_{\xi
_n(j',l')+1}\bigl(j',l'\bigr)
\\
&&{}+ \sum
_{i=1}^{n} t_{j'l'}(\widehat{\boldsymbol {
\gamma}}_{i-1},\mathbf{S}_{i-1})\mathbh{1}_{\{
Z_i=(t_{j'},w_{l'})\}}
\\
&&{} - \sum_{i=1}^{\xi_n(j',l')+1} t_{j'l'}(
\widehat{\boldsymbol {\gamma}}_{i-1},\mathbf{S}_{i-1})
\mathbh{1}_{\{
Z_i=(t_{j'},w_{l'})\}}. %
\end{eqnarray*}
Moreover, since $\widetilde{N}_{\xi_n(j',l')}(j',l') \leq N_{\xi
_n(j',l')}(j',l') t_{j'l'}(\widehat{\boldsymbol{\gamma}}_{\xi
_n(j',l')},\mathbf{S}_{\xi_n(j',l')})$, then
\begin{eqnarray*}
&&\widetilde{N}_n\bigl(j',l'
\bigr) - N_n\bigl(j',l'\bigr)
t_{j'l'}(\widehat{\boldsymbol {\gamma}}_{n},
\mathbf{S}_{n})
\\
&&\quad \leq  M_n\bigl(j',l'
\bigr) -M_{\xi
_n(j',l')+1}\bigl(j',l'\bigr)+1
\\
 &&\qquad {} + \Biggl(N_{\xi_n(j',l')}\bigl(j',l'\bigr)
t_{j'l'}(\widehat{\boldsymbol {\gamma}}_{\xi_n(j',l')},
\mathbf{S}_{\xi_n(j',l')}) - \sum_{i=1}^{\xi_n(j',l')+1}t_{j'l'}(
\widehat{\boldsymbol{\gamma }}_{i-1},\mathbf{S}_{i-1})
\mathbh{1}_{\{Z_i=(t_{j'},w_{l'})\}} \Biggr)
\\
&&\qquad  {} - \Biggl( N_n\bigl(j',l'\bigr)
t_{j'l'}(\widehat{\boldsymbol{\gamma }}_{n},
\mathbf{S}_{n}) - \sum_{i=1}^{n}t_{j'l'}(
\widehat {\boldsymbol{\gamma}}_{i-1},\mathbf{S}_{i-1})
\mathbh{1}_{\{
Z_i=(t_{j'},w_{l'})\}} \Biggr). %
\end{eqnarray*}
Since $p_{jl}>0$, then as $n\rightarrow\infty$
\[
N_n(j,l)\rightarrow\infty\quad\mbox{and}\quad\frac{M_n}{N_n(j,l)}
\rightarrow0\qquad\mbox{a.s.}\qquad\forall j =0,\ldots,J; l = 0, \ldots,L.
\]
Moreover, as $n\rightarrow\infty$ at least one of $\widetilde
{N}_n(j',l')$ and $[N_n(j',l')-\widetilde{N}_n(j',l')]$ tends to
infinity a.s. Therefore $\widehat{{\boldsymbol{\gamma
}}}_n\rightarrow\boldsymbol{\gamma}$ a.s. and, from (\ref
{ipotesicov}), $\mathbf{S}_n\rightarrow\boldsymbol{\varsigma}$
a.s. Thus, as $n\rightarrow\infty$
\begin{eqnarray*}
t_{j'l'}(\widehat{\boldsymbol{\gamma}}_{n},
\mathbf{S}_{n}) - \frac{ \sum_{i=1}^{n} t_{j'l'}(\widehat{\boldsymbol{\gamma
}}_{i-1},\mathbf{S}_{i-1}) \mathbh{1}_{\{Z_i=(t_{j'},w_{l'})\}}} {
\sum_{i=1}^{n} \mathbh{1}_{\{Z_i=(t_{j'},w_{l'})\}}} \rightarrow0 \qquad
\mbox{a.s.}
\end{eqnarray*}
Furthermore, as $n\rightarrow\infty$
\begin{eqnarray*}
&& t_{j'l'} (\widehat{\boldsymbol{\gamma}}_{\xi
_n(j',l')},
\mathbf{S}_{\xi_n(j',l')} ) \frac{N_{\xi
_n(j',l')}(j',l')}{N_{n}(j',l')}
\\
&&\quad{} - \frac{ \sum_{i=1}^{\xi_n(j',l')+1}
t_{j'l'}(\widehat{\boldsymbol{\gamma}}_{i-1},\mathbf{S}_{i-1})
\mathbh{1}_{\{Z_i=(t_{j'},w_{l'})\}} }{ \sum_{i=1}^{n} \mathbh{1}_{\{
Z_i=(t_{j'},w_{l'})\}}}
\rightarrow0 \qquad\mbox{a.s.}
\end{eqnarray*}
and therefore $\lim_{n\rightarrow\infty} [\uppi_n(j',l') -
t_{j'l'}(\widehat{\boldsymbol{\gamma}}_{n},\mathbf
{S}_{n}) ]^+ =0$ a.s.

Analogously, $\lim_{n\rightarrow\infty} \{[1-\uppi_n(j',l')] -
[1-t_{j'l'}(\widehat{\boldsymbol{\gamma}}_{n},\mathbf
{S}_{n}) ] \}^+ = 0$ a.s. and thus
%
%
\begin{equation}
\label{conv max} \lim_{n\rightarrow\infty}\uppi_n
\bigl(j',l'\bigr)= t_{j'l'}(\boldsymbol {
\gamma },\boldsymbol{\varsigma}) \qquad\mbox{a.s.}
\end{equation}
{\spaceskip=0.2em plus 0.05em minus 0.02em Since $\exists!$ $\mathbf{t} (\widehat{\boldsymbol{\gamma
}}_{n},\mathbf{S}_{n}  )=[t_{jl}(\widehat{\boldsymbol
{\gamma}}_{n},\mathbf{S}_{n}) \dvt j =0,\ldots,J; l = 0, \ldots
,L]$ which is continuous and }$\boldsymbol{\varphi}(\mathbf{t}
(\boldsymbol{\gamma},\boldsymbol{\varsigma}
);\allowbreak  \boldsymbol{\gamma},\boldsymbol{\varsigma})=\mathbf{t}
(\boldsymbol{\gamma},\boldsymbol{\varsigma}  )$, then
from (\ref{conv max}) follows that
\[
\lim_{n\rightarrow\infty}\uppi_n(j,l) = t_{jl}(
\boldsymbol{\gamma },\boldsymbol{\varsigma}) \qquad\mbox{a.s.}\qquad\mbox{for
every } (j,l)\neq\bigl(j',l'\bigr)
\]
and Theorem~\ref{thm4} follows directly.
\end{appendix}

\section*{Acknowledgements}

We are grateful to the referees and the associate editor for their
comments and suggestions, which led to a substantially improved version
of the paper.


%

\printhistory

\end{document}